\documentclass[a4paper]{amsart}

\usepackage{amsmath,amssymb}
\usepackage[dvipsnames]{xcolor}
\usepackage{pictexwd}

\usepackage{graphicx}
\usepackage{hyperref}
   
\newtheorem{theorem}{Theorem}[section]
\newtheorem{proposition}[theorem]{Proposition}
\newtheorem{corollary}[theorem]{Corollary}
\newtheorem{lemma}[theorem]{Lemma}

\newtheorem{preremark}[theorem]{Remark}
\newtheorem{predefinition}[theorem]{Definition}
\newtheorem{preexample}[theorem]{Example}
\newtheorem{prenotation}[theorem]{Notation}
\newtheorem{preconjecture}[theorem]{Conjecture}
\newtheorem{assumption}[theorem]{Assumption}

\newenvironment{remark}{\begin{preremark}\rm}{\end{preremark}}
\newenvironment{definition}{\begin{predefinition}\rm}
{\end{predefinition}}

\allowdisplaybreaks
\newcommand{\ZZ}{{\mathbb{Z}}}
\newcommand{\NN}{{\mathbb{N}}}
\newcommand{\FF}{{\mathbb{F}}}

\newcommand{\X}{{\mathbb{X}}}
\newcommand{\PP}{{\mathbb{P}}}
\newcommand{\m}{\mathfrak{m}}
\renewcommand{\Im}{{\mathop{\rm Im}}}
\newcommand{\JJ}{{\mathfrak{J}}}
\newcommand{\DD}{{\mathcal{D}}}
\newcommand{\II}{{\mathcal{I}}}

\let\epsilon=\varepsilon

\def\phi{{\varphi}}
\let\Psi=\varPsi
\let\Phi=\varPhi
\let\theta=\vartheta
\let\rho=\varrho

\newcommand{\wt}{\mathop{\rm wt}\nolimits}
\newcommand{\GL}{\mathop{\rm GL}\nolimits}
\newcommand{\Mat}{\mathop{\rm Mat}\nolimits}
\newcommand{\Hom}{\mathop{\rm Hom}\nolimits}
\newcommand{\gHom}{\underline{\mathop{\rm Hom}\nolimits}}
\newcommand{\HF}{\mathop{\rm HF}\nolimits}
\newcommand{\Soc}{\mathop{\rm Soc}\nolimits}
\newcommand{\Ann}{\mathop{\rm Ann}\nolimits}
\newcommand{\sepdeg}{\mathop{\rm sepdeg}\nolimits}

\def\LT{\mathop{\rm LT}\nolimits}

\def\Hom{\mathop{\rm Hom}\nolimits}

\def\rk{\mathop{\rm rk}\nolimits}

\def\TTo#1{\mathop{\longrightarrow}\limits ^{#1}}

\def\tfrac #1#2{{\textstyle\frac{#1}{#2}}}
\def\tsum_#1^#2{{\textstyle\sum\limits_{#1}^{#2}}}
\def\tprod_#1^#2{{\textstyle\prod\limits_{#1}^{#2}}}

\definecolor{red}{rgb}{1.0, 0.0, 0.0}

\def\cocoa{\mbox{\rm
  C\kern-.13em o\kern-.07 em C\kern-.13em o\kern-.15em A}}
\def\apcocoa{\mbox{\rm
A\kern-0.13em p\kern -0.07em C\kern-.13em o\kern-.07 em C\kern-.13em
o\kern-.15em A}}

\begin{document}

\title{Code Equivalence, Point Set Equivalence, and Polynomial Isomorphism}

\author{Martin Kreuzer}
\address{Fakult\"at f\"ur Informatik und Mathematik, Universit\"at
Passau, D-94030 Passau, Germany}
\email{Martin.Kreuzer@uni-passau.de}

\date{\today}

\begin{abstract}
The linear code equivalence (LCE) problem is known to be equivalent to
the point set equivalence (PSE) problem, i.e., the problem to check whether two sets
of points in a projective space over a finite field differ by a linear change of coordinates.
For such a point set~$\X$, let $R$ be its homogeneous coordinate ring
and $\JJ_\X$ its canonical ideal. Then the LCE problem is shown to be equivalent to
an algebra isomorphism problem for the doubling $R/\JJ_\X$. As this doubling is
an Artinian Gorenstein algebra, we can use its Macaulay inverse system
to reduce the LCE problem to a Polynomial Isomorphism (PI) problem for homogeneous polynomials.
The last step is polynomial time under some mild assumptions about the codes. Moreover, for
indecomposable iso-dual codes we can reduce the LCE search problem to the 
PI search problem of degree~3 by noting that the
corresponding point sets are self-associated and arithmetically Gorenstein, so that
we can use the isomorphism problem for the Artinian reductions of the coordinate rings
and form their Macaulay inverse systems.
\end{abstract}



\maketitle

%
%

\section{Introduction}

Given two linear $[n,k]_q$-codes $C$ and~$C'$ with generator matrices $G$ and~$G'$, respectively, 
the \textit{Linear Code Equivalence (LCE)} problem asks whether there exist a matrix $A\in \GL_k(\FF_q)$, 
a diagonal matrix $D\in \GL_n(\FF_q)$, and a permutation matrix $P \in \GL_n(\FF_q)$ such that 
$G' = A\cdot G\cdot D\cdot P$, and if this is the case, to find the matrices $A$, $D$, and~$P$.
This problem lies at the heart of the security of many
code-based cryptographic primitives, for instance the McEliece cryptosystem (see~\cite{McE}), the Niederreiter 
cryptosystem (see~\cite{Nie}), or the recent LESS signature scheme (see~\cite{BMPS}).
Although it has been studied extensively, its precise complexity seems to be unknown. 
For an overview of the main results about this topic we refer to~\cite{BBPS},
\cite{PS}, and~\cite{BMi}, and in particular to~\cite{Bea} 
for the relation to other isomorphism problems.
For special classes of codes, some structural attacks
are able to solve the LCE problem, for instance when the dimension of the hull of~$C$ is small (see~\cite{BMS}).
However, for self-dual codes, or, more generally, iso-dual codes, the problem has been considered to be hard.

In this paper we study reductions of the LCE problem to other problems, in particular to the
Polynomial Isomorphism (PI) problem. Further polynomial time reductions of this problem considered before were
a reduction to the lattice isomorphism problem (see~\cite{BW}) and to the permutation equivalence problem
(see~\cite{SS}). Notice, however, that the latter reduction is not polynomial time in the sense
considered here, because we require polynomial time in $(n,k,\log(q))$, whereas~\cite{SS}
only provides a polynomial time algorithm in $(n,k,q)$. Moreover, in~\cite[Prop.~3.3]{GQ2}
the authors construct a reduction of the LCE problem to the 3-Tensor Isomorphism (3-TI) problem
which has been shown to reduce to the PI problem in degree 3 over fields satisfying certain assumptions 
on their characteristic (see~\cite{AS1}, \cite{AS2} and~\cite{GQ1}). Note that these reductions result in the study of
cubic polynomials in $O(n^2)$ indeterminates, rather than the ones in~$k$ indeterminates considered here. 
A better reduction from 3-TI to 3-PI follows from~\cite{GQ3}, but it still yields
cubic polynomials involving a linear blow-up of the number of indeterminates.

The approach we take here is to
view the columns of the generator matrices~$G$ and~$G'$ as sets of points~$\X$ and~$\X'$ in the projective space
$\PP^{k-1}$ over~$\FF_q$ and to apply algebraic geometry methods from the theory of 0-dimensional
subschemes of projective spaces. The first observation is that it is easy to reduce the LCE problem
to the case when~$C$ and~$C'$ are projective codes, i.e., when no two columns in~$G$ or in~$G'$
are $\FF_q$-linearly dependent. Then the associated point sets~$\X$ and~$\X'$ in~$\PP^{k-1}$,
given by the columns of~$G$ and~$G'$, respectively, consist of~$n$ points each.

Our second observation is that the LCE problem is polynomial time equivalent to the
\textit{Point Set Equivalence (PSE)} problem which asks whether there exists a linear change
of coordinates $\Lambda:\; \PP^{k-1} \longrightarrow \PP^{k-1}$ such that $\Lambda(\X)=\X'$.
This seems to be folklore knowledge (see e.g.~\cite{KMW}), but for the convenience of the reader
we spell out an explicit proof in Prop.~\ref{prop:CEPandPSEP}.
Here ``polynomial time'' refers to the size of the input parameters $n$, $k$, and~$\log(q)$.

Many properties of a projective linear code~$C$ are reflected in algebraic properties of objects
related to the associated point set~$\X$, in particular its homogeneous vanishing ideal $I_\X$ in
$P=\FF_q[x_1,\dots,x_n]$ and its homogeneous coordinate ring $R=P/I_\X$.
Another algebraic object we employ is the canonical module~$\omega_R$ of~$R$. 
This is a finitely generated, graded $R$-module which is strongly related 
to the dual code of~$C$. Formally, it is defined as follows.
We may assume that there exists a linear non-zero\-divisor $\ell \in R_1$. 
Then the canonical module of~$R$ is defined as $\omega_R = \gHom_{\FF_q[\ell]}(R,\FF_q[\ell])(-1)$. 

The properties that the code~$C$ is indecomposable,
or that the set~$\X$ is indecomposable (i.e., not contained in a union of disjoint linear spaces),
are equivalent to the property that~$\omega_R$ is generated in degrees $\le -1$ 
(see Prop.~\ref{prop:indecomp}). Moreover, the property that~$C$ is
iso-dual, i.e., equivalent to its dual, is equivalent to the property that~$\X$ is self-associated,
i.e., equal to its so-called Gale transform, a classical concept in the study of finite sets of points
in projective spaces (see Prop.~\ref{prop:IsoDualSelfAssoc}). If we combine this and look
at an indecomposable iso-dual code~$C$, we get what is known as an arithmetically Gorenstein scheme~$\X$ 
with the additional property that its Hilbert function satisfies
$\Delta\HF_\X:\; 1\;k{-}1\; k{-}1\;1$. 
Equivalently, these are sets of points whose canonical module~$\omega_R$
is graded free of rank~1 and its Hilbert function satisfies $\HF_{\omega_R}(-2)=1$ (see Thm.~\ref{thm:CharIsoDual}).

To be able to use the canonical module for studying the LCE problem, we imbed it as an ideal 
into the homogeneous coordinate ring~$R$ of~$\X$. Different such embeddings are possible 
(see Prop.~\ref{prop:CanId}), but the simplest and
most natural one is the canonical ideal $\JJ_{R/\FF_q[\ell]}$. This is a homogeneous ideal
of~$R$ which starts in degree $r_\X = \min \{ i\ge 0 \mid \HF_R(i)=n\}$. The number $r_\X$ is an important 
invariant of~$\X$ and is called the (Hilbert) regularity index of~$\X$.
The canonical ideal can be computed explicitly in polynomial time (see~Prop.~\ref{prop:CompCanId}). 
Moreover, in the case of an iso-dual code, its generator has a very simple description (see Prop.~\ref{prop:CanIdOfIsoDual}).

The key property for our purposes is that the residue class ring $D_\X = R/\JJ_{R/\FF_q[\ell]}$
is a graded Artinian Gorenstein ring with socle degree $2r_\X-1$. This ring is called the
doubling of~$\X$ and its main use for our purposes is revealed by Theorem~\ref{thm:EquivDoubling}:
the LCE and PSE problems are equivalent to the property that there exists
a linear change of coordinates which induces an $\FF_q$-algebra isomorphism $\lambda_D:\;
D_{\X} \longrightarrow D_{\X'}$.

In order to check the latter property, we introduce our final tool, the Macaulay inverse system.
Starting with the divided power algebra $\DD = \gHom_{\FF_q}(P, \FF_q) = \FF_q[\pi_1,\dots,\pi_k]$, 
where $P=\FF_q[x_1,\dots,x_k]$ and $\pi_i$ is the projection to~$x_i$, we consider the action~$\circ$ of~$P$ 
on~$D$ by contraction and define for every homogeneous ideal $J\subseteq P$
its Macaulay inverse system as $J^\perp = \{ \psi \in \DD \mid f\circ\psi = 0$ for all $f\in J\}$.

By a theorem of F.S.\ Macaulay, the inverse system of an Artinian Gorenstein algebra with socle degree~$d$
is a principal ideal generated by a homogeneous polynomial of degree~$d$ (see Thm.~~\ref{thm:MacInverse})
and vice versa. Since this duality is equivariant with respect to the action of $\GL_k(\FF_q)$ by 
linear coordinate changes (see Prop.~\ref{prop:EquivarMacInv}), we reduce the LCE and PSE 
problem to the \textit{Polynomial Isomorphism (PI)} problem for the Macaulay inverse polynomials $\Phi_{\X}$
and $\Phi_{\X'}$ of~$\X$ and~$\X'$, respectively. The latter problem is to check whether these two homogeneous polynomials
of degree $2r_\X-1$ differ only by a homogeneous linear change of coordinates, and if this is the case, to find the
matrix defining that linear change of coordinates (see Thm.~\ref{thm:EquivReduce}).

The calculation of the Macaulay inverse polynomial $\Phi_\X$
of the ideal $\widehat{J}_\X \subseteq P$ defining the doubling $D_\X = P / \widehat{J}_\X$
is made explicit in Thm.~\ref{thm:CompMacInv}.
Although this is not a polynomial time algorithm in general, it is so if we have a bound on the regularity
index of~$\X$. For instance, if~$\X$ is associated to a linear code with a good data rate and
minimal distance, such bounds are available (see Rem.~\ref{rem:RegBound}).

Another setting in which we have a polynomial time reduction is the case 
of indecomposable iso-dual codes, or equivalently, of indecomposable self-associated point sets.
Namely, we can use the fact that the Artinian reduction $\overline{R} = R / \langle \ell\rangle$
of the homogeneous coordinate ring of~$\X$ is an Artinian Gorenstein ring with socle degree~3.
Hence the computation of the corresponding Macaulay inverse system is achieved in polynomial time,
and we have reduced the LCE search problem to the PI search problem in degree 3
(see Prop.~\ref{prop:ArithGorMacInverse} and Cor.~\ref{cor:ArithGorPEQ}). Thus, in the setting 
of iso-dual codes, we do not need any bound on the regularity index and we end up
with a PI search problem of degree~3 which is lower than the degree~5 that 
we have to deal with after using the doubling of~$\X$.

The PI search problem in degree~3 has been suggested as the security certificate 
of the IP1S identification scheme by J.\ Patarin (see~\cite{Pat}). This scheme has been extensively
cryptoanalyzed and is not considered secure anymore if randomly generated
instances are used (see~\cite{BFFP}, \cite{Kay1}, \cite{Kay2}).
 
As far as the original question about the worst-case complexity of the LCE problem is concerned, 
we can say that it is known to be not
NP-hard (unless the polynomial hierarchy collapses, see~\cite{PR}), 
and that it reduces to the computation of (special) Macaulay inverse polynomials
and the PI problem. The latter is known to be at least as hard as the tensor isomorphism problem
which in turn is considered much harder than the graph-isomorphism problem (see~\cite{PGC}).
For the special case of iso-dual codes, we have reduced the LCE search problem in polynomial time
to the PI search problem in degree~3, which is also known as the cubic equivalence problem, 
has been studied intensely, and is considered to be a hard problem.

Unless explicitly noted otherwise, the notation in this paper follows the books~\cite{KR1} 
and~\cite{KR2}. For basic definitions and facts about set of points in projective spaces,
we refer the readers to~\cite[Sect.~6.3]{KR2}. The author has strived to make this paper readable for
specialists in coding theory, complexity theory, and algebraic geometry. Therefore many concepts are
explained in greater detail than may be necessary, and the paper has been made as self-sufficient as
possible.

\bigskip\bigbreak
%
%

\section{Linear Code Equivalence and Point Set Equivalence}\label{sec2}

In this paper we work with linear codes defines over a finite field~$\FF_q$,
where $q=p^e$ is a power of a prime number~$p$.
A \textit{linear $[n,k,d]_q$-code} is a linear code of length $n\ge 1$ and dimension 
$k = \dim_{\FF_q}(C)\ge 1$ with \textit{minimal distance} $d = \min\{ \wt(c) \mid c\in C \setminus \{0\}\}$. 
Here $\wt(c)$ denotes the \textit{Hamming weight} of a code word $c\in C$. 
In other words, the set~$C$ is a $k$-dimensional $\FF_q$-vector subspace of $\FF_q^n$.
If the minimal distance is not known or of no relevance, we also say that $C$ is an $[n,k]_q$-code.

When we put the coordinates of the elements of a basis of~$C$ into the rows of a
matrix $G \in \Mat_{k,n}(\FF_q)$, we obtain a \textit{generator matrix} of~$C$.
Of course, the generator matrix is not unique: if we choose another basis of~$C$, i.e., 
if we multiply~$G$ from the left with an element of $\GL_k(\FF_q)$, we obtain another
generator matrix of the same code.

Two linear $[n,k]_q$-codes $C,C'$ in $\FF_q^n$ are called \textit{isometric}
if there exists a bijective $\FF_q$-linear map $\phi:\; C \longrightarrow C'$
such that~$\phi$ preserves the Hamming metric, i.e.\ such that $\wt(\phi(c))=\wt(c)$
for all $c\in C$. The following famous result of Florence MacWilliams classifies
these isometries (see~\cite{Wil}).

\begin{theorem}[MacWilliams Extension Theorem]$\mathstrut$\\
Let $C,C'$ be two linear $[n,k]_q$-codes. Then~$C$ and~$C'$ are (linearly)
isometric if and only if there exist a permutation matrix $P \in \GL_n(\FF_q)$ 
and an invertible diagonal matrix $D\in \GL_n(\FF_q)$ such that $C' = \{ c\,D\,P \mid
c\in C\}$. 
\end{theorem}

In this case we say that~$C$ and~$C'$ are \textit{linearly equivalent} and write
\hbox{$C\sim C'$}. Sometimes~$C$ and~$C'$ are also called \textit{monomially equivalent} or simply
\textit{equivalent}. It is clear that linearly equivalent codes have the same
length and the same dimension. Using their generator matrices $G$ and $G'$, respectively, 
we can detect whether $C\sim C'$ by checking whether $G' = A\, G\, D\, P$
for a matrix $A\in \GL_k(\FF_q)$, a diagonal matrix $D\in \GL_n(\FF_q)$,
and a permutation matrix $P\in  \GL_n(\FF_q)$.
Notice that we do not include automorphisms of~$\FF_q$ in our definition of code equivalence,
as these are normally not used in the application scenarios we are interested in and lead to
\textit{semi-linearly equivalent codes}. Thus we arrive at the following important problems.

\medskip
\noindent{\bf Linear Code Equivalence Decision Problem.} Decide whether the linear codes generated by~$G$ and~$G'$
are linearly equivalent, i.e., whether there exist an invertible matrix $A\in \GL_k(\FF_q)$, a
diagonal matrix $D\in \GL_n(\FF_q)$, and a permutation matrix $P\in \GL_n(\FF_q)$ 
such that $G' = A \cdot G \cdot D \cdot P$.
\smallskip

\noindent{\bf Linear Code Equivalence Search Problem.} Given two linear codes with generator matrices~$G$
and~$G'$ which are linearly equivalent, find a matrix $A\in \GL_k(\FF_q)$, a diagonal matrix 
$D\in \GL_n(\FF_q)$, and a permutation matrix $P\in \GL_n(\FF_q)$ 
such that $G' = A \cdot G \cdot D \cdot P$.
\smallskip

\noindent Usually, this search problem is simply called the {\bf Linear Code
Equivalence (LCE)} problem.
\smallskip

\medskip
Next we recall the relation between linear codes and sets of points in the
$(k-1)$-dimensional projective space~$\PP^{k-1}$ over~$\FF_q$.
The homogeneous coordinate ring of~$\PP^{k-1}$ will be denoted by $P=\FF_q[x_1,\dots,x_k]$.

Given a set of~$n$ $\FF_q$-rational points $\X = \{p_1,\dots,p_n\}$ 
in $\PP^{k-1}$, we let $I_{\X} \subseteq P$
be its \textit{homogeneous vanishing ideal}. This is the ideal generated by
all homogeneous polynomials $f \in P$ such that
$f(p_1) = \cdots = f(p_n)=0$. Notice that the condition whether or not~$f$ vanishes at a
point~$p_i$ does not depend on the choice of a particular tuple of homogeneous
coordinates of~$p_i$.

The ring $R = P / I_{\X}$ is called the \textit{homogeneous coordinate ring}
of~$\X$. Since~$I_\X$ is a homogeneous ideal, the ring~$R$ is non-negatively
graded, i.e., we have $R = \bigoplus_{i\ge 0} R_i = \bigoplus_{i\ge 0} 
P_i  /(I_\X)_i$, where $P_i$ and $R_i$ denote the vector spaces
of homogeneous elements of degree~$i$.

Let us fix a homogeneous coordinate tuple $p_j = (p_{j1} : \dots : p_{jk})$
for every point $p_j\in\X$. Then every element $f \in R_i$ of~$R$
has a well-defined value $f(p_j)= f(p_{j1},\dots,p_{jk})$ at~$p_j$.
For every $i\ge 1$, we now define a map 
$$
\epsilon_i:\; R_i \longrightarrow (\FF_q)^n
\hbox{\quad\rm by letting\quad} \epsilon_i(f) = (f(p_1),\dots,f(p_n)).
$$
Then $\epsilon_i$ is an injective $\FF_q$-linear map, and its image
$C_i = \epsilon_i(R_i)$ is called the \textit{$i$-th evaluation code}
(or the \textit{$i$-th generalized Reed-Muller code}) associated to~$\X$.

On the other hand, given a linear $[n,k]_q$-code~$C$ with
generator matrix~$G$, we let~$\X$ be the set of points in~$\PP^{k-1}$
defined by the non-zero columns of~$G$. 
It is called the \textit{associated set of points} corresponding to~$G$.
Notice that~$\X$ may consist of less than~$n$ points. 
The following case is particularly useful here.

\begin{definition}
Let $C$ be a linear $[n,k]_q$-code, and let $G \in \Mat_{k,n}(\FF_q)$
be a generator matrix for~$C$. If no two columns of~$G$ are $\FF_q$-linearly
equivalent, then~$C$ is called a \textit{projective} linear code.
\end{definition}

Notice that this condition is independent of the choice of the generator matrix~$G$
and implies that no column of~$G$ is zero. Clearly, a linear $[n,k]_q$-code~$C$
is projective if and only if the associated set of points~$\X$ in~$\PP^{k-1}$
consists of~$n$ distinct points.
Every projective linear code is an evaluation code in the following manner.

\begin{remark}
Let $C$ be a projective linear $[n,k]_q$-code with 
generator matrix $G\in \Mat_{k,n}(\FF_q)$.
For $i=1,\dots,n$, let $(p_{i1},\dots,p_{ik})^{\rm tr}$
be the $i$-th column of~$G$, where $p_{ij}\in \FF_q$. Consider the set of points 
$\X = \{p_1,\dots,p_n\}$ in $\PP^{k-1}$ given by $p_i=(p_{i1} : \cdots : p_{ik})$.
Since any two columns of~$G$ are $\FF_q$-linearly independent, the set~$\X$ consists of~$n$
distinct points.

Then~$C = \epsilon_1(R_1)$ equals the first evaluation code of~$\X$. To see why 
this is true, it suffices to note that $\rk(G)=k$ implies that $\X$ spans $\PP^{k-1}$,
whence $R_1 = P_1 = \FF_q x_1 \oplus \cdots \oplus \FF_q x_k$, and that the evaluations
of the coordinate functions $x_1,\dots,x_k$ at the points of~$\X$ yield exactly
the matrix~$G$.
\end{remark}

The following two lemmas show that the hypothesis that the codes under consideration 
are projective linear codes is no serious restriction
for the study of the LCE problem. (Similar lemmas are shown in~\cite[Lemmas 3.4 and~3.5]{CSV}.)

\begin{lemma}\label{lemma:NoZeroColumn}
Let $C,C'$ be two linear $[\hat{n},k]_q$-codes in $(\FF_q)^{\hat{n}}$, and let $G,G' \in \Mat_{k,\hat{n}}(\FF_q)$
be generator matrices for~$C$ and~$C'$, respectively. Let $\overline{G}$ and $\overline{G'}$
be obtained by deleting the zero columns in~$G$ and~$G'$, respectively, and let $\overline{C},\overline{C}'$
be the codes generated by these matrices. Then $C \sim C'$ holds if and if $\overline{C}\sim \overline{C'}$.
\end{lemma}

The proof is analogous to the proof of~\cite[Lemma~3.5]{CSV}.

\begin{lemma}\label{lemma:ReduceG}
Let $C,C'$ be two linear $[\hat{n},k]_q$-codes in $(\FF_q)^{\hat{n}}$, and let $G,G' \in \Mat_{k,\hat{n}}(\FF_q)$
be generator matrices for~$C$ and~$C'$, respectively.
\begin{enumerate}
\item[(a)] For $j = 1, \dots, \hat{n}$, delete the $j$-th column of~$G$ if it is $\FF_q$-linearly dependent
on one of the previous columns of~$G$. Let $n$ be the number of remaining columns, let $\overline{G}$
be the remaining matrix, and let $\overline{C} \subseteq \FF_q^n$ be the linear code with 
generator matrix~$\overline{G}$. Then~$G$ and~$\overline{G}$ have the same associated set of
points~$\X$, and this set consists of~$n$ points.

\item[(b)] Let us perform the same operation as in~(a) with~$G'$. We obtain a code $\overline{C}'$ with a 
generator matrix $\overline{G}' \in \Mat_{k,n'}(\FF_q)$. If~$C$ and~$C'$ are linearly equivalent, then
$n=n'$ and the code $\overline{C}$ is linearly equivalent to~$\overline{C}'$.

\item[(c)] Let us construct the matrices $\overline{G}$ and $\overline{G}'$ as in (a),(b).
If the codes $\overline{C}$ and~$\overline{C'}$ are linearly equivalent, then the codes~$C$ and~$C'$
are linearly equivalent and the matrices defining this linear equivalence can be computed from those
for $\overline{C} \sim \overline{C}'$ in polynomial time.

\end{enumerate}
\end{lemma}

\begin{proof}
Claim~(a) follows from the observation that two columns of~$G$ are $\FF_q$-linearly
equivalent if and only if they yield the same point in~$\PP^{k-1}$. 

Next we prove~(b). As in~\cite[Lemma~3.4]{CSV}, we see that the number of columns of~$G$ which correspond 
to a certain point in~$\PP^{k-1}$ equals the number of columns of~$G'$ with this property. This implies $n=n'$,
and that the columns which are canceled in~$G$ and~$G'$ correspond to each other one to one. Therefore,
if we consider $G' = A \cdot G \cdot D \cdot P$ and delete in the matrix $D\,P$ the row-column pairs corresponding
to the canceled columns, we get $\overline{G}' = A\cdot \overline{G}\cdot \overline{D\;P}$, and this yields the claim.

Finally, we show~(c). As before, the columns which have been canceled in~$G$ and~$G'$ correspond to each 
other one to one. By multiplying~$G$ and~$G'$ on the right with suitable permutation matrices and
arguing by induction, we may assume that~$G$ and~$G'$ arise from~$\overline{G}$ and $\overline{G}'$, respectively, 
by appending a column which is a non-zero multiple of the last column. By multiplying~$G$ and~$G'$ on the right
by a suitable invertible diagonal matrix, we may even assume that the last two columns of both~$G$ and~$G'$ are
equal. Starting from $\overline{G}' = A \cdot \overline{G} \cdot \overline{D} \cdot \overline{P}$, it now 
sufficies to append the row $(0,\dots,0,1)$ and the column $(0,\dots,0,1)^{\rm tr}$ to~$\overline{D}$
and to~$\overline{P}$ in order to get matrices $D,\, P$ such that $G' = A \cdot G \cdot D \cdot P$, as required. 
\end{proof}

In view of these two lemmas we may assume that the linear codes for which we want to study the LCE
problem are projective.
Using the correspondence between projective linear codes and set of points, the
LCE problem can be reformulated using the following problem.
\smallskip

\noindent{\bf Point Set Equivalence Decision Problem.} Given two sets of points $\X,\X'$
in the projective space~$\PP^{k-1}$ over $\FF_q$, decide whether there exists
a linear change of coordinates $\Lambda:\; \PP^{k-1} \longrightarrow \PP^{k-1}$
such that $\Lambda(\X)=\X'$. In this case we say that~$\X$ and~$\X'$ are {\it equivalent} point sets
and write $\X \sim \X'$.
\smallskip

\noindent{\bf Point Set Equivalence Search Problem.} Given two equivalent point sets~$\X$
and~$\X'$ in~$\PP^{k-1}$, find a linear change of coordinates $\Lambda:\; \PP^{k-1} \longrightarrow \PP^{k-1}$
such that $\Lambda(\X)=\X'$.
\smallskip

\noindent Usually, this search problem is simply called the \textbf{Point Set Equivalence (PSE)} problem. 

Equivalently, we can ask for an $\FF_q$-algebra isomorphism $\lambda:\; R \longrightarrow R'$
between the homogeneous coordinate rings $R = P/I_\X$ and
$R' = P/I_{\X'}$ which is defined by homogeneous linear polynomials.
\smallskip

Now we have the following equivalence between the LCE and PSE problems.

\begin{proposition}[Linear Code Equivalence and Point Set Equivalence]\label{prop:CEPandPSEP}$\mathstrut$\\
Let $C,C'$ be two projective linear $[n,k]_q$-codes with generating matrices $G,G' \in \Mat_{k,n}(\FF_q)$,
and let $\X,\X'$ be their associated sets of points. 

Then the LCE decision problem and the PSE decision problem are equivalent.
Moreover, also the LCE search problem and the PSE search problem are equivalent.
\end{proposition}

\begin{proof}
Multiplying a generator matrix~$G$ from the right by a permutation matrix amounts
to changing the order of the associated points. Multiplying it from the right by an invertible 
diagonal matrix amounts to changing the tuples of homogeneous coordinates representing the individual points.
Both operations do not change the associated set of points~$\X$. 
Lastly, we note that multiplying~$G$ from the left by an invertible matrix of size~$k$
amounts to a linear change of coordinates in~$\PP^{k-1}$.
\end{proof}

In this paper we will frequently mention algorithms which are of polynomial time.
Here \textit{polynomial time} refers to a polynomial bound in terms of the
input parameters~$k$, $n$, and $\log(q)$ of the code~$C$ or the associated
point set~$\X$. 
Both directions of the equivalences in the preceding proposition
can be carried out in polynomial time as follows.

\begin{remark}\label{rem:CEPvsPSEP}
Let $C,C'$ be two projective linear $[n,k]_q$-codes with generator matrices 
$G, G' \in \Mat_{k,n}(\FF_q)$, and let $\X,\X'$ be their associated sets of points. 
\begin{enumerate}
\item[(a)] Assume that~$C$ and~$C'$ are linearly equivalent, i.e., that
there exist a matrix $A\in \GL_k(\FF_q)$, a diagonal matrix 
$D\in \GL_n(\FF_q)$, and a permutation matrix $P\in \GL_n(\FF_q)$ 
such that $G' = A \cdot G \cdot D \cdot P$.
Then the associated sets of points $\X = \{ p_1, \dots, p_n \}$ and $\X' = \{ p'_1, \dots, p'_n \}$ 
are equivalent via the homogeneous linear change of coordinates $\Lambda:\; \PP^{k-1} \longrightarrow \PP^{k-1}$
given by $\Lambda(p_i) = A\cdot p_{\pi(i)}^{\rm tr}$, where $\pi$ is the permutation of~$P$.

\item[(b)] Conversely, assume that $\X \sim \X'$, and that the corresponding linear change of coordinates
$\Lambda:\; \PP^{k-1} \longrightarrow \PP^{k-1}$ is given by $\Lambda(p_i) = A\cdot p_i^{\rm tr}$.
Let us normalize the points of~$\X$ such that their first non-zero homogeneous coordinate is one.
Thus we write $A \cdot p_i^{\rm tr} = c_i \hat{p}_i$ with $c_i \in \FF_q \setminus \{0\}$ for $i=1,\dots,n$.
Moreover, we similarly normalize the points of~$\X'$ and write $p'_j = c'_j \hat{p}'_j$
with $c'_j \in \FF_q \setminus \{0\}$ for $j=1,\dots,n$. Using at most $n$ comparisons, we find
for every $i\in \{1,\dots,n\}$ the unique index $j\in \{1,\dots,n\}$ such that $\hat{p}_i = \hat{p}'_j$.
Hence we obtain a permutation matrix $P\in \GL_n(\FF_q)$ with $\hat{p}_i P = \hat{p}'_j$ for all corresponding
pairs $(i,j)$. Finally, the diagonal matrix $D\in \GL_n(\FF_q)$ such that $G' = A \cdot G \cdot D \cdot P$
is obtained from the factors~$c_i$ and~$c'_j$. All these operations are clearly performed in polynomial time.
They show that the codes~$C$ and~$C'$ are linearly equivalent.

\end{enumerate}
\end{remark}

In the subsequent sections, the following notion will be useful.

\begin{definition}
Let $R$ be a non-negatively graded, finitely generated $\FF_q$-algebra.
The map $\HF_R:\; \mathbb{Z} \longrightarrow \mathbb{N}$ defined by $\HF_R(i) = \dim_{\FF_q}(R_i)$
for all $i\ge 0$ and by $\HF_R(i)=0$ for $i<0$ is called the \textit{Hilbert function} of~$R$.
\end{definition}

If~$R$ is the homogeneous coordinate ring of a set of points in $\PP^{k-1}$, we call it also
the Hilbert function of~$\X$ and write $\HF_\X$ instead of $\HF_R$. The Hilbert function
of a set of $n$ points in $\PP^{k-1}$ has the following properties (see~\cite[Prop.~6.3.23]{KR2}).
\begin{enumerate}
\item[(1)] $\HF_\X(i)=0$ for $i<0$ and $\HF_\X(0)=1$.

\item[(2)] There exists a number $r_\X$, called the \textit{regularity index} of~$\X$,
such that $1 = \HF_X(0) < \HF_\X(1) < \cdots < \HF_\X(r_{\X})=n$.

\item[(3)] $\HF_\X(i) = n$ for all $i\ge r_\X$.
\end{enumerate}

For the further examination of the associated point set~$\X$,
we introduce the following hypothesis.

\begin{assumption}\label{ass:nzd}
Given a finite set of $\FF_q$-rational points $\X$ in~$\PP^{k-1}$, we assume that
there exists a linear form $L \in P_1$ which defines a 
hyperplane $H = \mathcal{Z}(L)$ such that no point of~$\X$ is contained in~$H$.
\end{assumption}

In other words, we assume that the set~$\X$ is not a \textit{blocking set} in~$\PP^{k-1}\!$.
By~\cite[Thm.~2]{BB}, a blocking set consists of at least $\frac{q^{k-1}-1}{q-1}$ points.
For most of the associated point sets of the codes in this paper, the above assumption is
satisfied. In order to study the PSE search problem, we may assume that the assumption
is satisfied by the following lemma.

\begin{lemma}
Let $\X,\X'$ be finite sets of $\FF_q$-rational points in~$\PP^{k-1}$ which consist
of~$n$ points each.
\begin{enumerate}
\item[(a)] There exist a finite extension field~$K$ of $\FF_q$ and a homogeneous linear
polynomial $L \in K[x_1,\dots,x_n]_1$ such that, if we consider $\X$ and $\X'$ as sets of points
in $\PP^{k-1}_K$, no point of~$\X$ or~$\X'$ lies on the hyperplane $H=\mathcal{Z}(L)$.

\item[(b)] Let $L$ be chosen as in~(a). Then the PSE search problem for $\X,\X'$
in $\PP^{k-1}$ is equivalent to the PSE search problem for $\X,\X'$ in $\PP^{k-1}_K$.

\end{enumerate}
\end{lemma}

\begin{proof}
To show~(a), we let $f\ge 2$ and choose $K=\FF_{q^f}$. For large enough~$f$,
the set $\X\cup\X'$ is not a blocking set in~$\PP^{k-1}_K$ by~\cite[Thm.~2]{BB}.

Next we prove~(b). Since~$\X$ and~$\X'$ span $\PP^{k-1}$ by the definition of~$k$,
we can choose~$k$ points $p_1,\dots,p_k$ in~$\X$ such that the matrix $(p_1\vert \cdots\vert p_n)$ 
is in $\GL_k(\FF_q)$. Let $\Lambda:\; \PP^{k-1}_K \longrightarrow \PP^{k-1}_K$ be a 
homogeneous linear change of coordinates which solves the PSE search problem
for $\X,\X'$ in~$\PP^{k-1}_K$. Since it maps~$\X$ to~$\X'$,
it sends $(p_1,\dots,p_k)$ to a tuple of points $(p'_1,\dots,p'_k)$ of~$\X'$ 
for which the matrix $(p'_1 \vert \cdots \vert p'_n)$ is in $\GL_k(\FF_q)$.
Consequently, the entries of the matrix defining~$\Lambda$ are in~$\FF_q$, i.e.\ 
the map~$\Lambda$ is defined over $\FF_q$ and solves the PSE search problem for $\X,\X'$
over~$\FF_q$.
\end{proof}

Notice that, under a base field extension as in this lemma, the relevant invariants 
of the code and the point set, such as the parameters of the code, 
or the Hilbert function of the coordinate ring of the point set, do not change.
\smallskip

Next we extend $\{L \}$ to an $\FF_q$-basis $\{L, y_1, \dots, y_{k-1}\}$ of 
the vector space of linear forms $P_1$. Let $\overline{P}=\FF_q[y_1,\dots,y_{k-1}]$,
and let $\eta:\; P \longrightarrow \overline{P}$ be
the $\FF_q$-algebra epimorphism induced by letting $\eta(L)=0$. In degree one,
this is the projection to $\FF_q y_1 \oplus \cdots \oplus
\FF_q y_{k-1}$ along $\FF_q\,L$. Then the ring $\overline{R} = \overline{P} / 
\eta(I_\X) = R/\langle \ell\rangle$, where~$\ell$ is the residue class of~$L$
in~$R$, is called the \textit{Artinian reduction} of~$R$.
It is a 0-dimensional, local, graded Artinian $\FF_q$-algebra whose maximal ideal
is the ideal $\m = \langle \bar{y}_1,\dots, \bar{y}_{k-1} \rangle$ generated by the residue classes
of the indeterminates. Instead of $\eta(I_\X)$ we also write $\bar{I}_\X$.

\begin{remark}
The Hilbert function of $\overline{R} = R/\langle \ell\rangle$ is the (first) difference 
function of~$\HF_R$, because $\ell\in R_1$ is a homogeneous non-zerodivisor.
In other words, we have
$$
\HF_{\overline{R}}(i) \;=\; \Delta \HF_R(i) \;=\; \HF_R(i) - \HF_R(i-1)
\hbox{\quad for all\quad} i\in\ZZ.
$$
\end{remark}

Now we define the following property of the set of points~$\X$.

\begin{definition}
Let $\X = \{p_1,\dots,p_n\}$ be a set of $\FF_q$-rational points in $\PP^{k-1}$ as above.
We say that~$\X$ is \textit{arithmetically Gorenstein} if~$R$ is a Gorenstein ring.
\end{definition}

Here the ring~$R$ is called a Gorenstein ring if and only if~$\overline{R}$ is a Gorenstein ring,
and the latter property is defined by the condition that the \textit{socle} 
$\Soc(\overline{R}) = \Ann_{\overline{R}}(\m)$ of~$\overline{R}$ is a 1-dimensional $\FF_q$-vector space.
It is well-known that this condition does not depend on the choice of~$\ell$.

The following characterization of arithmetically Gorenstein sets of points was first shown
in~\cite{DGO}. It was generalized to 0-dimensional schemes over algebraically closed fields
in~\cite{Kre1} and to 0-dimensional schemes over arbitrary fields in~\cite{KLR}.

\begin{proposition}\label{prop:CharArithGor}
A set of points $\X = \{p_1,\dots,p_n\}$ in $\PP^{k-1}$ as above
is arithmetically Gorenstein if and only if the following properties are satisfied.
\begin{enumerate}
\item[(a)] The Hilbert function of~$\X$ is {\rm symmetric}, that is
$\Delta \HF_\X(i) = \Delta \HF_\X( r_\X-i )$ for $i=0,\dots,r_\X$.  

\item[(b)] The set of points~$\X$ has the {\rm Cayley-Bacharach property}, that is,
for $i=1,\dots,n$, the set $\mathbb{Y}_i = \X \setminus \{p_i\}$ satisfies
$\HF_{\mathbb{Y}_i}(r_\X-1)= \HF_\X(r_\X-1)$.

\end{enumerate}
\end{proposition}

The Cayley-Bacharach property can also be formulated using the following notion. 
For $i\in \{1,\dots,n\}$, a polynomial $f\in P$ is called a \textit{separator}
(or sometimes an \textit{indicator function})
of~$p_i$ in~$\X$ if $f(p_i)=1$ and $f(p_j)=0$ for $j\ne i$ (see also~\cite{GKR}
or~\cite[Sect.~6.3]{KR2}).
Here we use the coordinate tuple of~$p_i$ with respect to the affine
space $D_+(L)$ to define $f(p_i)$. More precisely, we use the coordinate system
$(L,y_1,\dots,y_{k-1})$ of~$\PP^{k-1}$ and the homogeneous coordinates 
$p_i = (1: \bar{p}_{i1} : \cdots \bar{p}_{i\, k{-}1})$ with respect to this coordinate system.

Clearly, the separator of a point in~$\X$ is not unique, as it can be 
modified by adding an element of~$I_\X$. The smallest degree of a separator
of~$p_i$ in~$\X$ is called the \textit{separator degree} of~$p_i$ in~$\X$
(or simply the \textit{degree} of~$p_i$ in~$X$, or the \textit{v-number} of~$p_i$)
and denoted by $\sepdeg_{\X}(p_i)$. 

\begin{remark}\label{rem:sepdeg}
For $i=1,\dots,n$, we have $\sepdeg_{\X}(p_i) \le r_\X$.
The set of points~$\X$ has the Cayley-Bacharach property if and only if
$\sepdeg_{\X}(p_i)=r_\X$ for $i=1,\dots,n$
(see~\cite{GKR} or~\cite[Tutorial 88]{KR2}).
\end{remark}

Some basic data of~$\X$ can be computed in polynomial time as follows.

\begin{remark}\label{rem:BuMo}
Given a finite set of $\FF_q$-rational points $\X = \{ p_1,\dots,p_n\}$ in $\PP^{k-1}$ 
through the coordinates of the points $p_1,\dots,p_n$, the following data can be computed
in polynomial time.
\begin{enumerate}
\item[(a)] The Buchberger-M\"oller Algorithm (cf.~\cite{BM} or~\cite[Thms.~2.7 and~3.12]{ABKR})
computes, for a given term ordering~$\sigma$, the reduced $\sigma$-Gr\"obner basis of~$I_\X$
and a set of terms $\mathcal{O}_\sigma(I_\X) = \mathbb{T}^k \setminus \LT_\sigma(I_\X)$
which represents an $\FF_q[\ell]$-basis of~$R$.

\item[(b)] If we use a degree compatible term ordering~$\sigma$ in~(a), the set~$\mathcal{O}_\sigma(I_\X)$ also
yields the Hilbert function $\HF_\X$ and the regularity index~$r_\X$ of~$\X$ in polynomial time.

\item[(c)] A suitably modified version of the Buchberger-M\"oller Algorithm also computes
the separators of~$\X$ and the separator degrees of the points of~$\X$ (cf.~\cite[Cor.~2.8 and Rem.~3.15]{ABKR}).

\end{enumerate}
For detailed complexity estimates, we refer the readers to~\cite{MMM} and~\cite[Sect.~4]{ABKR}.
\end{remark}

\bigskip\bigbreak
%
%

\section{The Canonical Module}

As before, we let $\X = \{p_1,\dots, p_n\}$ be a set of $\FF_q$-rational 
points in~$\PP^{k-1}$. Recall that we let $R=\FF_q[x_1,\dots,x_k]/I_\X$ be the homogeneous 
coordinate ring of~$\X$ and that we assumed that there exists a linear form
$L\in \FF_q[x_1,\dots,x_k]_1$ such that $\X \cap \mathcal{Z}(L) = \emptyset$.

In the following we denote the residue class of~$L$ in~$R_1$ by~$\ell$.
Then~$\ell$ is a non-zerodivisor of~$R$, because it is not contained in any
minimal prime ideal of~$R$. The minimal primes of~$R$ are the ideals defining the
points~$p_i$. They are also the associated prime ideals of~$R$, as~$R$ is reduced.
Since $\dim(R)=1$ and $\ell\in R_1$ is a non-zerodivisor, it follows that~$R$ is a 
Cohen-Macaulay ring (see also~\cite[Prop.~6.3.21]{KR2}).
Thus $\FF_q[\ell]$ is a polynomial ring inside the 1-dimensional Cohen-Macaulay ring~$R$.
Hence $\FF_q[\ell] \subseteq R$ is a Noether normalization and~$R$
is a finitely generated free $\FF_q[\ell]$-module. To compute an $\FF_q[\ell]$-basis of~$R$,
it suffices to compute an $\FF_q$-basis of the Artinian reduction $\overline{R} = R/\langle \ell\rangle$
and lift it. The following module plays a central role in this paper.

\begin{definition}(The Canonical Module)$\mathstrut$\\
Consider the graded $\FF_q[\ell]$-module $\omega_R = \gHom_{\FF_q[\ell]}(R, \FF_q[\ell])(-1)$,
where $\gHom$ denotes the module generated by the homogeneous $\FF_q[\ell]$-linear
maps and $(-1)$ denotes a degree shift by~$-1$.
This $\FF_q[\ell]$-module is turned into an $R$-module by defining
$$
a \cdot \phi(b) \;=\; \phi(ab) \hbox{\quad for\quad} \phi\in \omega_R \hbox{\quad and\quad} a,b\in R.
$$
The resulting $R$-module is called the \textit{canonical module} of~$R$.
\end{definition}

The canonical module of~$R$ is a finitely generated, graded $R$-module. It plays an
important role in commutative algebra, for instance
in graded local Grothendieck duality theory, it is a dualizing module in the 
category of finitely generated graded $R$-modules and uniquely determined up to isomorphism.
Its Hilbert function $\HF_{\omega_R}(i) = \dim_{\FF_q}((\omega_R)_i)$ satisfies
$$
\HF_{\omega_R}(i) \;=\; n - \HF_R(-i) \hbox{\quad for all\quad} i\in\ZZ.
$$
(cf.~\cite[Prop.~3.1]{GKR}). The canonical module $\omega_R$ can be computed in polynomial time as follows.

\begin{remark}\label{rem:CompCanMod}
In the above setting, a presentation of the canonical module~$\omega_R$
can be calculated as follows (see~\cite[Sect.~4.5]{KR3}).
From the $\FF_q[\ell]$-basis $\mathcal{O}_\sigma(\bar{I}_\X) = \{t_1,\dots,t_n\}$ of~$R$
we calculate, in polynomial time, the multiplication matrices of~$x_1,\dots,x_k$
in this basis. Their transposed matrices are the multiplication matrices
of the canonical module $\omega_R = \FF_q[\ell]\, t_1^\ast \oplus \cdots\oplus
\FF_q[\ell]\, t_n^\ast$. 

Next we use these matrices to derive a presentation
$\omega_R \cong \bigoplus_{i=1}^n R\, e_i / U$ where~$U$
is generated by elements of the form $x_i e_j - \sum_{\nu=1}^n c_\nu e_\nu$
with $c_\nu \in \FF_q[\ell]$. Lastly, we minimize this presentation
using the graded version of Nakayama's Lemma (see~\cite[Prop.~1.7.15]{KR1}).

Clearly, all these computations can be achieved in
polynomial time. An improved algorithm for computing $\omega_R$
starts with the projections to the \textit{socle monomials} 
instead of $\{t_1^\ast, \dots, t_n^\ast\}$ (see~\cite{BK}).
\end{remark}

Next we also introduce the notions of indecomposable codes and sets of points.

\begin{definition}
Let $C$ be an $[n,k]_q$-code, and let $\X$ be a set of $n$ $\FF_q$-rational points in $\PP^{k-1}$.
\begin{enumerate}
\item[(a)] The code~$C$ is called the \textit{direct sum} of two codes~$C_1$ and~$C_2$
if it has a generator matrix of the form $G = \binom{G_1 \;\;0\;}{\;0\;\; G_2}$
where $G_i$ is a generator matrix for~$C_i$ for $i=1,2$.

\item[(b)] If~$C$ is not the direct sum of two codes, it is called \textit{indecomposable}.

\item[(c)] The set of points $\X$ is called \textit{decomposable} if $\X$ is a disjoint union 
$\X = \X_1 \cup \X_2$ such that there exist linear spaces $L_i$ in~$\PP^{k-1}$
with $\X_i \subseteq L_i$ for $i=1,2$ and with $L_1\cap L_2 = \emptyset$.

\item[(d)] If the set of points~$\X$ is not decomposable, it is called \textit{indecomposable}.
\end{enumerate}
\end{definition}

The indecomposability of a code can be checked using its associated set of points as follows.

\begin{proposition}\label{prop:indecomp}
Let $C$ be a projective $[n,k]_q$-code with a generator matrix~$G$, let~$\X$ be the associated set of points of~$G$, 
and let $R$ be the homogeneous coordinate ring of~$\X$. Then the following conditions are equivalent.
\begin{enumerate}
\item[(a)] The code~$C$ is indecomposable.

\item[(b)] The set of points~$\X$ is indecomposable.

\item[(c)] The canonical module $\omega_R$ is generated in degrees $\le -1$.
\end{enumerate}
\end{proposition}

\begin{proof}
First we show that~(a) implies~(b). Suppose that~$\X$ is decomposable. Write
$\X=\X_1 \cup \X_2$ with $\X_i \subseteq L_i$ for $i=1,2$ and linear spaces
$L_i \subseteq \PP^{k-1}$ such that $L_1\cap L_2 = \emptyset$. The latter condition
implies $\dim(L_1) + \dim(L_2) \le k-2$. By enlarging $L_1$ or~$L_2$ if necessary, 
we may assume that we have equality here. Then we choose a coordinate system
$x'_1,\dots,x'_k$ of~$\PP^{k-1}$ such that $x'_1,\dots,x'_d$ is a coordinate system of~$L_1$,
where $\dim(L_1)=d-1$,
and $x'_{d+1},\dots,x'_k$ is a coordinate system of~$L_2$. Here we may assume that
$L_1 = \mathcal{Z}(x'_{d+1},\dots,x'_k)$ and $L_2 = \mathcal{Z}(x'_1, \dots, x'_d)$.
In this coordinate system, the matrix whose columns are the homogeneous coordinate tuples
of the points of~$\X$ is block diagonal. Therefore the code~$C$ is decomposable, in contradiction
to the hypothesis. 

The implication (b)$\Rightarrow$(a) follows analogously. If some generator matrix~$G$
is block diagonal, then the homogeneous coordinate tuples of the points of~$\X$
with respect to the corresponding coordinate system of~$\PP^{k-1}$ show that~$\X$
is decomposable.

The equivalence of~(b) and~(c) was shown in~\cite[Prop.~5.5]{EP}. Notice that the set
of points~$\X$ is non-degenerate, since $\rk(G)=k$.
\end{proof}

\bigskip\bigbreak
%
%

\section{Iso-Dual Codes and Self-Associated Sets of Points}\label{sec4}

Continuing in the above setting, we let $C$ be a projective linear $[n,k]_q$-code
with a generator matrix $G\in \Mat_{k,n}(\FF_q)$.
On the $\FF_q$-vector space $\FF_q^n$, we have the \textit{standard bilinear form}
$\Phi:\; \FF_q^n \times \FF_q^n  \longrightarrow \FF_q$ defined by
$$
\Phi((a_1,\dots,a_n),(b_1,\dots,b_n))= a_1 b_1 + \cdots + a_n b_n.
$$
Since it is non-degenerate, we can form the \textit{dual code}
$$
C^\perp \;=\; \{(c_1,\dots,c_n) \in \FF_q^n \mid a_1c_1 + \cdots + a_n c_n = 0
\hbox{\ \rm for all\ }(a_1,\dots,a_n) \in C\}
$$
This is an $[n,n-k,d^\perp]_q$-code, where the minimal distance $d^\perp$
appears to be difficult to determine.

\begin{definition}
Let $C$ be an $[n,k,d]_q$-code.
\begin{enumerate}
\item[(a)] The code~$C$ is called \textit{iso-dual} if $C\sim C^\perp$.

\item[(b)] The code $C$ is called \textit{self-dual} if $C=C^\perp$.

\item[(c)] The code $C$ is called \textit{weakly self-dual} if $C \subseteq C^\perp$.

\end{enumerate}
\end{definition}

Since the dimension of the dual of an $[n,k,d]_q$-code is $n-k$, it follows that
the length of an iso-dual code satisfies $n=2k$.

Our goal in this section is to reinterpret the dual code of~$C$ and the concept of iso-dual
codes in terms of the associated sets of points. Let~$G$ be a generator matrix
of~$C$ and $\X = \{p_1,\dots,p_n\}$ the associated set of $\FF_q$-rational points
in~$\PP^{k-1}$ defined by the columns of~$G$. The associated sets of points of~$C$
and $C^\perp$ are related using the following construction (see~\cite[Def.~1.1]{EP}).

\begin{definition}(The Gale Transform)$\mathstrut$\\
Let $\X$ be a set of~$n=k+\ell$ $\FF_q$-rational points in~$\PP^{k-1}$, where $k,\ell \ge 1$.
Let $G\in \Mat_{k,n}(\FF_q)$ be a matrix whose columns are given by 
homogeneous coordinate tuples for the points of~$\X$.
A set $\X'$ of~$n$ $\FF_q$-rational points in $\PP^{\ell-1}$ is called
a \textit{Gale transform} of~$\X$ if the homogeneous coordinate tuples of
the points of~$\X'$ form a matrix $G' \in \Mat_{\ell,n}(\FF_q)$
such that $G'\, D\; G^{\rm tr} = 0$ for some invertible diagonal matrix
$D\in \GL_n(\FF_q)$.
\end{definition}

Here the matrix~$D$ expresses the fact that the choice of the homogeneous coordinates of a 
point is determined up to a non-zero scalar. The Gale transform of~$\X$ is
unique only up to a homogeneous linear change of coordinates in~$\PP^{\ell-1}$, since the
definition of~$G'$ merely requires that its rows generate the kernel of $D\, G^{\rm tr}$. 

The set of points~$\X$ is called \textit{self-associated} if~$\X$ is its own Gale transform.
In this case we have $k=\ell$ and $n=2k$. The first mathematician to study self-associated sets
of points was G.\ Castelnuovo~\cite{Cas} in 1889. The general Gale transform was introduced 
by A.B.\ Coble in~\cite{Cob1,Cob2,Cob3}. A very general and extensive discussion of the Gale transform
in given in~\cite{EP}. The following relation between the Gale transform and the dual code
follows immediately from the definitions.

\begin{proposition}\label{prop:IsoDualSelfAssoc}
Let $C$ be a projective linear $[n,k]_q$-code, let $\X$ be an associated set of points of~$C$, let $C^\perp$
be the dual code of~$C$, assume that also $C^\perp$ is projective, and let $\X^\perp$ be an associated set 
of points of~$C^\perp$.
\begin{enumerate}
\item[(a)] The set of points $\X^\perp$ is a Gale transform of~$\X$.

\item[(b)] The code~$C$ is iso-dual if and only if~$\X$ is self-associated.

\end{enumerate}
\end{proposition}

Using the associated set of points, iso-dual codes can be characterized
as follows.

\begin{theorem}[Characterization of Indecomposable Iso-Dual Codes]$\mathstrut$\label{thm:CharIsoDual}\\
Let $C$ be a projective $[n,k]_q$-code with generator matrix $G\in \Mat_{k,n}(\FF_q)$, and let~$\X$
by the associated set of points in~$\PP^{k-1}$. Then the following conditions
are equivalent.
\begin{enumerate}
\item[(a)] The code $C$ is indecomposable and iso-dual.

\item[(b)] The set of points~$\X$ is arithmetically Gorenstein
and satisfies $\Delta\HF_\X:\; 1\; k{-}1 \; k{-}1\; 1\; 0\; \dots$.

\item[(c)] The canonical module $\omega_R$ is a graded free $R$-module of rank~1
and satisfies $\HF_{\omega_R}(-2)=1$, $\HF_{\omega_R}(-1) = k$, $\HF_{\omega_R}(0)=2k-1$,
and $\HF_{\omega_R}(i)= 2k=n$ for $i\ge 1$.

\end{enumerate}
\end{theorem}

\begin{proof}
First we show that~(a) is equivalent to~(b). Since~$C$ is iso-dual, we know that $n=2k$.
As the rank of~$G$ is~$k$, the points of~$\X$ span~$\PP^{k-1}$. This implies
$\HF_\X(1) = k+1$. By Proposition~\ref{prop:indecomp}, the condition in~(a)
that~$C$ is indecomposable is equivalent to the condition that $\omega_R$ is generated
in degrees $\le -1$. Hence we can apply~\cite[Thm.~7.2]{EP}. It yields that~(a)
is equivalent to the condition that~$\X$ is arithmetically Gorenstein.
In particular, this forces the Hilbert function of~$\X$ to be symmetric,
and thus we get the claimed function $\Delta \HF_\X$.

The equivalence of~(b) and~(c) is well-known in commutative algebra (see~\cite[Satz 6.14]{HK}
or~\cite[Prop.~2.1.3]{GW}). In our setting, it also follows by combining Prop.~\ref{prop:CharArithGor}
and~\cite[Thm.~3.5]{GKR}. The stated Hilbert function of~$\omega_R$ follows from the above formula.
\end{proof}

Notice that the Hilbert function in part~(c) of this theorem is already fixed by
$\HF_{\omega_R}(-2)=1$, since~$\X$ is non-degenerate and consists of $2k$ points.
In particular, $\omega_R$ is graded free with a basis element in degree -2.

\bigskip\bigbreak
%
%

\section{The Canonical Ideal}
\label{sec5}

Once again we use the above setting. Let $C$ be a projective $[n,k]_q$-code with generator matrix
$G \in \Mat_{k,n}(\FF_q)$, and let $\X = \{p_1, \dots, p_n\}$ be the
associated set of points in~$\PP^{k-1}$. Since~$\X$ is a reduced
scheme, the canonical module~$\omega_R$ of the homogeneous coordinate
ring~$R$ of~$\X$ can be embedded as an ideal in~$R$ (see~\cite{HK}
and~\cite{Kre2}). More precisely, it can be embedded into the part of~$R$
which we can describe very explicitly as follows.

By Assumption~\ref{ass:nzd}, there exists a linear form $L\in P_1$
whose residue class $\ell \in R_1$ is a non-zerodivisor. 
Thus $\X$ is contained in the affine space $D_+(L)$, and using
this we can form well-defined functional values for elements of~$R$.

Recall that a homogeneous element $f_i\in R$ is called a separator
of~$p_i$ in~$\X$ if $f_i(p_i)=1$ and $f_i(p_j)=0$ for $j\ne i$.
As shown in~\cite{GKR}, there exists a separator $f_i \in R_{r_\X}$ for
every $i\in \{1,\dots,n\}$. These separators $f_1,\dots,f_n$ form an 
$\FF_q$-basis of~$R_{r_\X}$. Moreover, for every $j\ge 0$, we have
$$
R_{r_\X+j} \;=\; \FF_q \,\ell^j f_1 \oplus \cdots \oplus \FF_q\, \ell^j f_n
$$
and the multiplication in this part of the ring~$R$ is given by $g\cdot f_i = 
g(p_1)\, \ell^j f_1 + \cdots + g(p_n)\, \ell^j f_n$ for every $g\in R_j$.

Using the identification $\omega_R = \gHom_{\FF_q[\ell]}(R, \FF_q[\ell])(-1)$,
the canonical module of~$R$ can now be embedded as an ideal of~$R$, as the next
result shows (see~\cite[Prop.~1.9]{Kre2}).

\begin{proposition}[The Canonical Ideal of a Set of Points]\label{prop:CanId}$\mathstrut$\\
Consider the map $\Phi:\; \omega_R \longrightarrow R(2r_\X-1)$ which is defined as
follows. Let $j\ge 0$, and let $\phi\in (\omega_R)_{-r_\X+1+j} = 
\gHom_{\FF_q[\ell]}(R,\FF_q[\ell])_{-r_\X+j}$ be given by $\phi(f_i)=c_i \ell^j$
with $c_i\in \FF_q$ for $i=1,\dots,n$. If we let 
$$
\Phi(\phi) \;=\; \ell^j\,(c_1f_1 + \cdots + c_n f_n) \;=\; \phi(f_1)\, f_1 
+ \cdots + \phi(f_n)\, f_n
$$
then~$\Phi$ is an injective homomorphism of graded $R$-modules. 

In particular, 
$\mathfrak{J}_{R/\FF_q[\ell]} = \Im(\Phi)(-2r_\X+1)$ is a homogeneous ideal of~$R$
which is, up to a degree shift, isomorphic to~$\omega_R$. It is called the
{\rm canonical ideal} of~$R$.  
\end{proposition}

To simplify the notation, we write $\JJ_R$ for the canonical ideal of~$\X$ if~$\ell$ 
is clear from the context. From the description of the Hilbert function of~$\omega_R$, 
we immediately get the Hilbert function of $\mathfrak{J}_R$:
$$
\HF_{\JJ_R}(i) = 0 \hbox{\quad\rm for\;}i<r_\X \hbox{\quad\rm and\quad}
\HF_{\JJ_R}(r_\X+j) = n - \HF_R(r_\X-1-j) \hbox{\quad\rm for\;} j\ge 0.
$$
The canonical ideal of~$\X$ depends on the choice of~$\ell \in R_1$, of course.
Moreover, the canonical module can also be embedded into higher degree components of~$R$.
The next proposition describes these non-uniquenesses explicitly. A similar result
was shown in~\cite[Thm.~3.4]{Boi}.

\begin{proposition}
Let $\X$ be a set of points in $\PP^{k-1}$ as above. In particular, 
let $\ell\in R_1$ be a non-zerodivisor of~$R$.
\begin{enumerate}
\item[(a)] Given a further non-zerodivisor $\ell' \in R_1$ and the corresponding
embedding $\Phi'\!:\, \omega_R \longrightarrow R(2r_\X-1)$, we let
$\mathfrak{J}_{R/\FF_q[\ell']} = \Im(\Phi')(-2r_\X+1)$ be the canonical ideal of~$R$
with respect to~$\ell'$.

Then an element $g = \ell^i\, (c_1 f_1 + \cdots + c_n f_n) \in R_{r_\X+i}$ with 
$c_1,\dots,c_n\in K$ satisfies $g \in \mathfrak{J}_{R/\FF_q[\ell']}$
if and only if 
$$
\ell^i\, (\ell'(p_1)^{-2r_\X+1} c_1 f_1 + \cdots + \ell'(p_n)^{-2r_\X+1}
c_n f_n) \;\in\; \mathfrak{J}_{R/\FF_q[\ell]}.
$$
In other words, we have $\ell^{2r_\X-1} \mathfrak{J}_{R/\FF_q[\ell']} = 
(\ell')^{2r_\X-1} \mathfrak{J}_{R/\FF_q[\ell]}$.

\item[(b)] For every $i\ge 0$, the map
$$
\Phi^{(i)}:\; \omega_R \longrightarrow R(2r_\X-1-i)
$$
given by $\Phi^{(i)}(\phi) = \ell^i \Phi(\phi)$ for $\phi \in \omega_R$
is an injective homomorphism of graded $R$-modules and its image satisfies
$\Im(\Phi^{(i)})(-2r_\X+1+i) = \ell^i \, \mathfrak{J}_{R/\FF_q[\ell]}$.

In other words, for every degree $d = r_\X+i \ge r_\X$, the ideal $\ell^i \, \mathfrak{J}_{R/\FF_q[\ell]}$ 
is a homogeneous ideal of~$R$ which is isomorphic to the canonical ideal 
and which starts in degree~$d$.
\end{enumerate}
\end{proposition}

\begin{proof}
Claim~(a) was shown in~\cite[Prop.~1.8]{Kre2}. Claim~(b) is an immediate consequence of
the description of the homogeneous components and the multiplication of~$R$ in degrees
$\ge r_\X$ given above.
\end{proof}

Next we show how the canonical ideal of~$\X$ can be calculated in polynomial time.

\begin{proposition}[Computing the Canonical Ideal]\label{prop:CompCanId}$\mathstrut$\\
Let $R=P/I_\X$ be the homogeneous coordinate ring of a finite set of $n$ $\FF_q$-rational points~$\X$
in~$\PP^{k-1}$, and let $\ell\in R_1$ be the residue class of a linear form $L\in P_1$ which does not vanish at any
point of~$\X$.
\begin{enumerate}
\item[(a)] For $\phi\in \omega_R$ and $f\in R$, the element $\phi(f) \in \FF_q[\ell]$ can be computed
in polynomial time as follows.
\begin{enumerate}
\item[(1)] Extend $\{L\}$ to an $\FF_q$-basis $\{ L, y_1,\dots,y_{k-1} \}$ of 
$P_1$.

\item[(2)] Let $\bar{I}_\X$ be the residue class ideal of~$I_\X$ in
$\overline{P} = P/\langle L\rangle \cong \FF_q[y_1,\dots,y_{k-1}]$. Choose a term ordering~$\sigma$
on~$\overline{P}$ and compute $\mathcal{O}_\sigma(\bar{I}_\X) = \mathbb{T}^{k-1} \setminus
\LT_\sigma(\bar{I}_\X) = \{ t_1, \dots, t_n\}$.

\item[(3)] Using the $\FF_q[\ell]$-basis $\{t_1, \dots, t_n\}$ of~$R$, write $f= a_1 t_1 + \cdots 
+ a_n t_n$ with $a_i \in \FF_q[\ell]$.

\item[(4)] Using the $\FF_q[\ell]$-basis $\{ t_1^\ast, \dots, t_n^\ast \}$ of~$\omega_R$, write
$\phi = b_1 t_1^\ast + \cdots + b_n t_n^\ast$ with $b_j \in \FF_q[\ell]$.

\item[(5)]  Return $\phi(f) = a_1 b_1 + \cdots + a_n b_n$.
\end{enumerate}

\item[(b)] Find the separators $f_1,\dots,f_n \in R_{r_\X}$ as in Remark~\ref{rem:BuMo}.
Using the system of generators $\{\phi_1,\dots,\phi_n\}$ of~$\omega_R$
given by $\phi_i = t_i^\ast$ for $i=1,\dots,n$
(or a minimal system of generators computed as in Remark~~\ref{rem:CompCanMod}),
calculate $g_i = \phi_i(f_1) f_1 + \cdots + \phi_i(f_n) f_n$ for $i=1,\dots,n$.

Then $\JJ_{R/\FF_q[\ell]}$ is generated by $\{g_1,\dots,g_n\}$. The computation of this system
of generators is achieved in polynomial time.
\end{enumerate}
\end{proposition}

\begin{proof}
First we prove (a).
The set $\{t_1,\dots,t_n\}$ is an $\FF_q[\ell]$-basis of~$R$ by~\cite[Thm.~4.3.22]{KR2}.
To write~$f$ in terms of this basis, it suffices to compute the normal form of~$f$ with respect
to~$I_\X$ and the extension~$\overline{\sigma}$ of~$\sigma$ (see~\cite[Sect.~4.3]{KR2}).
Then $\{t_1^\ast,\dots, t_n^\ast\}$ is an $\FF_q[\ell]$-basis of $\omega_R = \gHom_{\FF_q[\ell]}
(R,\FF_q[\ell])(-1)$ and 
the claim follows from $t_i^\ast(t_j) = \delta_{ij}$ for $i,j=1,\dots,n$. 

Claim~(b) is a consequence of~(a) and Proposition~\ref{prop:CanId}.
\end{proof}

To analyze the canonical ideal in greater detail, we need one more ingredient.

\begin{remark}
Since~$\ell\in R_1$ is a non-zerodivisor for~$R$, the Artinian reduction 
$\overline{R} = R/ \langle \ell \rangle$
of~$R$ satisfies $\HF_{\overline{R}}(i) =\HF_R(i) - \HF_R(i-1) \ne 0$ if and only if 
$i \in \{0,\dots,r_\X\}$. The value $\Delta = \HF_{\overline{R}}(r_\X)$ is called the
\textit{last difference} of~$\HF_R$. Then the residue classes $\bar{f}_1,\dots,\bar{f}_n$
of the separators $f_i\in R_{r_\X}$ generate $\overline{R}_{r_\X}$, because $\{f_1,\dots,f_n\}$
is an $\FF_q$-basis of~$R_{r_\X}$ (see also~\cite[Prop.~2.13]{GKR}). 

Without loss of generality, we renumber the points of~$\X$ such that $\{\bar{f_1},\dots,
\bar{f_{\Delta}} \}$ is a $K$-basis of~$\overline{R}_{r_\X}$. For $j=\Delta_\X+1,\dots,n$,
we write $\bar{f}_{\Delta+j} = \beta_{j1} \bar{f}_1 + \cdots + \beta_{j \Delta} \bar{f}_\delta$
with $\beta_{jk}\in \FF_q$ for $j=1,\dots,n-\Delta$. Then we have
$$
f_{\Delta + j} - \beta_{j1} f_1 - \cdots - \beta_{j\Delta} f_\Delta \;=\; \ell\, g_j
$$
for some elements $g_j\in R_{r_\X-1}$, and the set $\{g_1,\dots,g_{n-\Delta}\}$ is an $\FF_q$-basis
of~$R_{r_\X-1}$ (see also~\cite[Lemma 1.2.b]{Kre2}). 
\end{remark}

Now we are ready to describe the canonical ideal of an iso-dual code.

\begin{proposition}[The Canonical Ideal of an Iso-Dual Code]\label{prop:CanIdOfIsoDual}$\mathstrut$\\
Let $C$ be an indecomposable iso-dual projective $[2k,k]_q$-code, let~$G$ be a generator matrix of~$C$, 
and let~$\X$ be the corresponding self-associated set of points in~$\PP^{k-1}$.
For $i=2,\dots,n$, write $\bar{f}_i = \beta_i \bar{f}_1$, where $\bar{f}_i = f_i + \langle\ell\rangle$
and $\beta_i \in \FF_q \setminus \{0\}$.

Then the element $\pi = f_1 + \beta_2 f_2 + \cdots +\beta_n f_n$ of $R_{r_\X}$
generates the canonical ideal $\JJ_{R/\FF_q[\ell]}$ which is a free $R$-module of rank~1.
\end{proposition}

\begin{proof}
By Theorem~\ref{thm:CharIsoDual}, the set of points~$\X$ is arithmetically Gorenstein
and $\omega_R$ is a graded free $R$-module of rank~1. Moreover, we have $r_\X=3$.
An element $\phi \in (\omega_R)_{-2}$
is an $R$-basis of~$\omega_R$ if and only if $\phi(f_i) \ne 0$ for $i=1,\dots,n$ 
(see~\cite[Lemma 3.4.c]{GKR}). Then the corresponding element $\pi = \phi(f_1) f_1 + \cdots 
+ \phi(f_n) f_n$ of~$(\JJ_{R/\FF_q[\ell]})_3$ is an $R$-basis of the graded free $R$-module 
$\JJ_{R/\FF_q[\ell]}$ of rank~1. As~$\X$ has the Cayley-Bacharach property, the degree of each point is~3,
and thus $f_i \notin \ell\, R_2$ for $i=1,\dots,n$. Here we have 
$\dim_{\FF_q}(\overline{R}_3) = \Delta \HF_\X(3)=1$ and we can write $\bar{f_i} = 
\beta_i \bar{f}_1$ with $\beta_i \in \FF_q \setminus \{0\}$ for $i=2,\dots,n$.
Thus, if we let $\phi(f_1)=1$ and $\phi(f_i)=\beta_i$ for $i=2,\dots,n$,
then $\phi(f_i - \beta_i f_1)=0$ implies $\phi(R_2)=0$ and $\phi$ is an $R$-basis of~$\omega_R$.
Consequently, the element $\pi = f_1 + \beta_2 f_2 + \cdots + \beta_n f_n$ is an $R$-basis 
of~$\JJ_{R/\FF_q[\ell]}$. 
\end{proof}

\bigskip\bigbreak
%
%

\section{Doublings of Sets of Points}
\label{sec6}

Continuing to use the notation of the preceding sections, we let $\X = \{p_1,\dots,p_n\}$
be a set of $\FF_q$-rational points in $\PP^{k-1}$ and $R=\FF_q[x_1,\dots,x_k]/I_\X$
its homogeneous coordinate ring. 
The next construction allows us to reduce the PSE problem to
a problem about 0-dimensional graded Gorenstein rings.

\begin{definition}
Let $\X$ be a set of points as above, let $\ell\in R_1$ be a non-zerodivisor,
let $\JJ_{R/\FF_q[\ell]}$ be the canonical ideal of~$R$, and let $i\ge 0$.
Then the graded $\FF_q$-algebra $D^{(i)}_\X = R / (\ell^i\, \JJ_{R/\FF_q[\ell]})$ is called the
$i$-th {\bf doubling} of~$\X$.

In the case $i=0$, we simply write $D_\X$ instead of $D^{(0)}_\X$ and call it the 
{\bf doubling} of~$\X$.
\end{definition}

The key property of the $i$-th doubling of~$\X$ is provided
by the following result.

\begin{proposition}[The Doubling of a Set of Points Is Gorenstein]$\mathstrut$\\
Let $\X$ be a set of points as above, let $\ell\in R_1$ be a non-zerodivisor,
let $\JJ_{R/\FF_q[\ell]}$ be the canonical ideal of~$R$, and let $i\ge 0$.
\begin{enumerate}
\item[(a)] The $i$-th doubling $D^{(i)}_\X = R/(\ell^i\, \JJ_{R/\FF_q[\ell]})$ 
is a 0-dimensional graded Gorenstein ring.

\item[(b)] The Hilbert function of $D^{(i)}_\X$ is given by
\begin{align*}
\HF_{D^{(i)}_\X}&(j) \qquad\qquad   = \; \HF_\X(j) \hbox{\quad\rm for\;} j<r_\X + i \hbox{\quad\rm and\;} \\
\HF_{D^{(i)}_\X}&(r_\X+i+j)  \; = \; \HF_\X(r_\X-1-j) \hbox{\quad \rm for\;} j \ge 0.
\end{align*}

\item[(c)] The (Hilbert) regularity index of $D^{(i)}_\X$ is $2r_\X+i-1$.

\end{enumerate}
\end{proposition}

\begin{proof}
For a proof of~(a), see \cite[Kor.~6.13]{HK} or~\cite[Prop.~3.3.18]{BH}. 
Claim~(b) follows from the Hilbert function of $\JJ_{R/\FF_q[\ell]}$ given in Section~\ref{sec5},
and~(c) follows immediately from~(b).
\end{proof}

For further properties of doublings, see for instance~\cite[Sect.~2.5]{KRS}.
In the case of self-associated sets of points, we obtain the following corollary.

\begin{corollary}[Doubling a Self-Associated Point Set]$\mathstrut$\\
Let $\X$ be a self-associated set of $2k$ $\FF_q$-rational points in $\PP^{k-1}$.
Then the doubling of~$\X$ is an Artinian Gorenstein ring $D_\X$ with Hilbert function
$$
\HF_{D_\X}:\;\; 1 \;\; k \;\; 2k-1 \;\; 2k-1 \;\; k \;\; 1,
$$
with $\dim_{\FF_q}(D_\X)=6k$, and with regularity index $r_{D_\X} = 5$.
\end{corollary}

The goal of this section is to reformulate the Point Set Equivalence Problem in several ways.
A linear change of coordinates $\Lambda:\; \PP^{k-1} \longrightarrow
\PP^{k-1}$ induces homomorphisms of several algebraic objects considered above in the
following ways.

\begin{remark}\label{rem:IndMaps}
Assume that a linear change of coordinates $\Lambda:\; \PP^{k-1} \longrightarrow \PP^{k-1}$
is given by a matrix $A = (a_{ij}) \in \GL_k(\FF_q)$. For a set of points
$\X = \{ p_1,\dots,p_n \}$ in~$\PP^{k-1}$, we let $\X' = \{ p'_1,\dots,p'_n \}$
where $p'_i = \Lambda(p_i)$ for $i=1,\dots,n$.
Furthermore, we denote the inverse matrix of~$A$ by $A^{-1} = (a'_{ij})$.
\begin{enumerate}
\item[(a)] The $\FF_q$-algebra isomorphism $\lambda:\; \FF_q[x_1,\dots,x_k] \longrightarrow
\FF_q[x_1,\dots,x_k]$ defined by $\lambda(x_i) = a'_{i1}x_1 + \cdots + a'_{ik} x_k$ 
for $i=1,\dots,k$ satisfies $\lambda(I_\X) = I_{\X'}$.

In particular, we get an induced $\FF_q$-algebra isomorphism $\lambda_R:\; R \longrightarrow
R'$, where $R' = \FF_q[x_1,\dots,x_k]/I_{\X'}$ is the homogeneous coordinate ring of~$\X'$.

\item[(b)] Let $L \in \FF_q[x_1,\dots,x_k]_1$ be a linear form such that $\X \cap \mathcal{Z}(L)
=\emptyset$. Then $L' = \lambda(L)$ satisfies $\X' \cap \mathcal{Z}(L')= \emptyset$.

\item[(c)] Let $\ell\in R_1$ be the residue class of~$L$, and let $\ell'\in R'_1$ be the
residue class of~$L'$, where $R'=\FF_q[x_1,\dots,x_k] / I_{\X'}$. Then the map~$\lambda_R$
induces an isomorphism of graded $R$-modules $\lambda_\omega:\;
\omega_R \longrightarrow \omega_{R'}$ which is given as follows.
For every $\phi \in \omega_R = \gHom_{\FF_q[\ell]}(R,\FF_q[\ell])$, the map
$\lambda_\omega(\phi)$ makes the following diagram commutative.
$$
\begin{matrix}
R & \TTo{\phi} & \FF_q[\ell] \\
{\scriptstyle\lambda_R} \downarrow \;\; && \;\;\downarrow {\scriptstyle\lambda_R} \\
R' & \TTo{\lambda_\omega(\phi)} & \FF_q[\ell']
\end{matrix}
$$

\item[(d)] In this setting, the map $\lambda_\omega$ induces an isomorphism
of graded $R$-modules $\lambda_\JJ:\; \JJ_{R/\FF_q[\ell]} \longrightarrow 
\JJ_{R'/\FF_q[\ell']}$ which is given as follows.
Up to a degree shift, let $\Phi$ and $\Phi'$ be the isomorphisms defined in Proposition~\ref{prop:CanId}.
Then $\lambda_\JJ$ is the $R$-module isomorphism which makes the following diagram
commutative.
$$
\begin{matrix}
\omega_R(-2r_\X+1) & \TTo{\lambda_\omega} & \omega_{R'}(-2r_\X+1) \\
{\scriptstyle\Phi}\downarrow\;  && \downarrow {\scriptstyle\Phi'} \\
\JJ_{R/\FF_q[\ell]} & \TTo{\lambda_\JJ} & \JJ_{R'/\FF_q[\ell']}
\end{matrix}
$$

\item[(e)] For every $i\ge 0$, the map $\lambda_\JJ$ induces an isomorphism
of graded $R$-modules $\lambda_{\JJ}^{(i)}:\; \ell^i\, \JJ_{R/\FF_q[\ell]}
\longrightarrow  (\ell')^i \, \JJ_{R'/\FF_q[\ell']}$.
\end{enumerate}
\end{remark}

By combining these isomorphisms with doublings of~$\X$, we obtain the following
reformulation of the point set equivalence problem.

\begin{theorem}[Characterizations of Point Set Equivalence]\label{thm:EquivDoubling}$\mathstrut$\\
Let $\X,\X'$ be two sets of $\FF_q$-rational points in $\PP^{k-1}$, each consisting of~$n$ points.
Assume that there exist linear forms $L,L'$ such that $\mathcal{Z}(L) \cap \X 
= \mathcal{Z}(L') \cap \X' = \emptyset$.
Let $R,R'$ be the homogeneous coordinate rings of~$\X$ and~$\X'$, respectively,
and let $\ell=L+I_\X \in R_1$ as well as $\ell' = L'+I_{\X'} \in R'_1$
be the residue classes of~$L$ and~$L'$.

Moreover, let $\Lambda:\; \PP^{k-1} \longrightarrow \PP^{k-1}$ be a linear change of
coordinates such that the induced map $\lambda: \FF_q[x_1,\dots,x_k] \longrightarrow
\FF_q[x_1,\dots,x_k]$ satisfies $\lambda(L)=L'$.
\begin{enumerate}
\item[(a)] If $\Lambda(\X) = \X'$ then the map $\Lambda$ induces for every $i\ge 0$ an
isomorphism of graded $\FF_q$-algebras $\lambda_{D,i}:\; D^{(i)}_{\X} \longrightarrow D^{(i)}_{\X'}$
which maps the residue class of a polynomial $f\in \FF_q[x_1,\dots,x_k]$ to the residue
class of $\lambda(f)$.

\item[(b)] The following conditions are equivalent.
   \begin{itemize}
   \item[(1)] $\Lambda(\X)=\X'$
   
   \item[(2)] For every $i\ge 0$, the map $\lambda$ induces an isomorphism of graded $\FF_q$-algebras
   $\lambda_{D,i}:\; D^{(i)}_{\X} \longrightarrow D^{(i)}_{\X'}$.
   
   \item[(3)] The map $\lambda$ induces an $\FF_q$-algebra isomorphism
   $\lambda_D:\; D_{\X} \longrightarrow D_{\X'}$.
   \end{itemize}

\end{enumerate}
\end{theorem}

\begin{proof}
To prove~(a), it suffices to show that $\lambda_{\JJ}(g) = \lambda_R(g)$
for every homogeneous element $g\in (\JJ_{R/\FF_q[\ell]})_{r_\X+d}$ with $d\ge 0$.
Then~$\lambda_R$ induces $\lambda_{D,0}$, and since 
$$
D^{(i)}_{\X} = R / (\ell^i\, \JJ_{R/\FF_q[\ell]}) \hbox{\rm \quad as well as \quad}
D^{(i)}_{\X'} = R' / ((\ell')^i\, \JJ_{R'/\FF_q[\ell']})
$$
the fact that $\lambda_R(\ell)=\ell'$ implies that~$\lambda_R$ induces an 
isomorphism $\lambda_{D,i}$ for every $i\ge 0$.

Next we let $f_1,\dots,f_n \in R_{r_\X}$ be the separators of degree~$r_\X$ of~$\X$, and we 
write $g=\ell^d\, (c_1 f_1 + \cdots + c_s f_s)$ with $c_1,\dots,c_s \in \FF_q$.
Then Proposition~\ref{prop:CanId} yields $g=\Phi(\phi)$, where $\phi \in (\omega_R)_{-r_\X+1+d}$
satisfies $\phi(f_j)= c_j \ell^d$ for $j=1,\dots,n$.

According to the diagram in Remark~\ref{rem:IndMaps}.d, we need to show 
$\Phi'(\lambda_\omega(\phi))(g) = \lambda_R(g)$. Notice that
$\Lambda(\X)=\X'$ implies that the elements $f'_j = \lambda_R(f_j)$ with
$j\in \{1,\dots,n\}$ are the separators of~$\X'$. Hence we get 
$\lambda_\omega(\phi)(f'_j) = \lambda_R (\phi(f_j)) = \lambda_R(c_j \ell^d) = c_j (\ell')^d$ 
for $j=1,\dots,n$. Consequently, it follows that
$$
\Phi'(\lambda_\omega(\phi))(g) \;=\; c_1 (\ell')^d f'_1 + \cdots + c_n (\ell')^d f'_n
\;=\; \lambda_R ( \ell^d (c_1 f_1 + \cdots + c_n f_n)) \;=\; \lambda_R(g)
$$
as desired.

Next we prove claim~(b). The implication (1)$\Rightarrow$(2) follows from~(a),
and the implication (2)$\Rightarrow$(3) is trivially true, as $\lambda_D = \lambda_{D,0}$.

Now we show that~(3) implies~(2). Let $i\ge 0$. By assumption, we have 
$\lambda_R(\JJ_{R/\FF_q[\ell]}) = \JJ_{R'/\FF_q[\ell']}$. Since $\lambda_R(\ell)=\ell'$,
it follows that $\lambda_R(\ell^i\, \JJ_{R/\FF_q[\ell]}) = (\ell')^i\, \JJ_{R'/\FF_q[\ell']}$
for every $i\ge 0$. Therefore the map~$\lambda_R$ induces an isomorphism 
$\lambda_{D,i}$ for every $i\ge 0$, as claimed.

Finally, we prove that~(2) implies~(1). We use $i=2$ and note that the initial degree
of $\ell^2\, \JJ_{R/\FF_q[\ell]}$ is $r_\X+2$. Therefore the hypothesis that $\lambda_{D,2}$
is an isomorphism implies that $(\lambda_R)_d:\; R_d \longrightarrow R'_d$ is an
isomorphism of $\FF_q$-vector spaces for every $d\le r_\X+1$. 
Consequently, we have $\lambda((I_\X)_d) = (I_{\X'})_d$ for every $d\le r_\X+1$.
In particular, we see that~$\X$ and~$\X'$ have the same Hilbert function and thus that
$r_{\X'} = r_\X$.
Now we apply the well-known fact that $I_\X$ and $I_{\X'}$ are generated in degrees $\le r_\X+1$
(see for instance~\cite[Prop.~1.1]{GM}) to conclude that $\lambda(I_\X)=I_{\X'}$. This yields $\Lambda(\X)=\X'$,
as was to be shown.
\end{proof}

\bigskip\bigbreak
%
%

\section{The Macaulay Inverse System}
\label{sec7}

The material in this section is valid over an arbitrary  base field~$K$. Later on, we will
be mostly interested in the case of a finite field $K=\FF_q$.
In the following we assume that~$J$ is homogeneous ideal in $P = K[x_1,\dots,x_k]$
such that $S = P/J$ is an Artinian ring. 
To define the Macaulay inverse system of~$S$, we first recall the divided
power algebra. For further details on this ring, we refer to~\cite[App.~A]{IK}.

Let $P = K[x_1,\dots,x_k]$ be equipped with the standard grading, and let 
$$
\DD \;=\; \gHom_K(P,K) \;=\; {\textstyle\bigoplus\limits_{j\ge 0}} \DD_{-j} 
\hbox{\; with\;} \DD_{-j} = \Hom_K(P_j,K)
$$
be its graded dual. Since the terms in $\mathbb{T}^n = \{ x_1^{\alpha_1} \cdots x_k^{\alpha_k}
\mid \alpha_i \ge 0\}$ form an $K$-basis of~$P$, the dual basis, i.e., 
the projections $\{ \pi^{[\beta]} \mid \beta=(\beta_1,\dots,\beta_k)\in \NN^k\}$ form an
$K$-basis of~$\DD$. Here we have $\pi^{[\beta]}(x^\alpha) = \delta_{\alpha \beta}$ for
$x^\alpha = x_1^{\alpha_1}\cdots x_k^{\alpha_k} \in \mathbb{T}^n$, as usual.

Letting $\pi_i= \pi^{[e_i]}$ with the standard basis vector $e_i = (0,\dots,0,1,0,\dots,0)$
for $i=1,\dots,k$, we define a product for power products of these elements via
$$
\pi^{[\beta]} \cdot \pi^{[\gamma]} \;=\; (\pi_1^{\beta_1} \cdots \pi_k^{\beta_k}) \cdot (\pi_1^{\gamma_1}
\cdots \pi_k^{\gamma_k}) \;=\; \tfrac{(\beta_1+\gamma_1)! \cdots (\beta_k+\gamma_k)!}{\beta_1 !
\cdots \beta_k !\, \gamma_1 ! \cdots \gamma_k !}\; \pi^{[\beta+\gamma]}
$$
for $\beta =(\beta_1,\dots,\beta_k)$ and $\gamma = (\gamma_1,\dots,\gamma_k)$ in $\NN^k$.
By extending this rule $K$-linearly, we obtain a $K$-algebra structure on~$\DD$,
and we call the resulting ring the \textit{divided power algebra} in~$k$ indeterminates over~$K$.

Next we introduce a left action $\circ:\, P \times \DD \longrightarrow \DD$ of~$P$ on~$\DD$ 
by extending the rule 
$$
x^\alpha \circ \pi^{[\beta]} = \begin{cases} 
0 & \hbox{if\;} \alpha_i>\beta_i \hbox{\; for some\; }i\in \{1,\dots,k\},\\
\pi^{[\beta-\alpha]} & \hbox{otherwise}
\end{cases}
$$
$K$-linearly. This is sometimes called the action by \textit{contraction}.
The following facts are well-known (see~\cite[App.~A]{IK}).

\begin{proposition}\label{prop:dividedpower}
Let $P = K[x_1,\dots,x_k]$ and $\DD = \gHom_K(P,K)$ as above.
\begin{enumerate}
\item[(a)] The multiplication defined above turns~$\DD$ into a commutative $K$-algebra.

\item[(b)] For every $j\ge 0$, the power products $\pi_1^{\beta_1} \cdots \pi_k^{\beta_k}$
with $\beta_1 + \cdots + \beta_k =j$ form a $K$-basis of~$\DD_{-j}$.

\item[(c)] The left action $\circ:\, P \times \DD \longrightarrow \DD$ turns~$\DD$
into a graded $P$-module.

\item[(d)] For every $j\ge 0$, the left action~$\circ$ induces a non-degenerate
bilinear map $\circ_j:\; P_j \times \DD_{-j} \longrightarrow K$.

\end{enumerate}
\end{proposition}

Now we extend the duality between~$P$ and~$\DD$ to residue class rings of~$P$
and $P$-submodules of~$\DD$ as follows.

\begin{definition}
Let $j\ge 0$, let $\phi\in \DD_{-j}$, and let~$J$ be a homogeneous ideal in~$P$.
\begin{enumerate}
\item[(a)] The graded ideal $\Ann_P(\phi) = \{f\in P \mid f \circ \phi = 0\}$
of~$P$ is called the \textit{annihilator ideal} (or the \textit{apolar ideal}) of~$\phi$.

\item[(b)] The $P$-submodule $J^\perp = \{\psi \in \DD \mid f \circ \psi = 0 \hbox{\; for all \;}
f\in J \}$ is called the \textit{Macaulay inverse system} of~$J$ (or of $P/J$).

\end{enumerate}
\end{definition}

In general, the Macaulay inverse system of a homogeneous ideal of~$P$ is not finitely
generated. However, if $S=P/J$ is an Artinian ring, as we assumed above, then $J^\perp$
is indeed a finitely generated graded $P$-submodule of~$\DD$. The most famous result
about Macaulay inverse systems is the following theorem of F.S.\ Macaulay 
(see~\cite[Ch.~IV]{Mac}; for a proof in the modern language of algebra, see~\cite[Lemma 2.14]{IK}).

\begin{theorem}[The Macaulay Inverse System of an Artinian Gorenstein Algebra]\label{thm:MacInverse}
Let $d\ge 1$. 
The maps $J \mapsto J^\perp$ and $\phi \mapsto \Ann_P(\phi)$ define a
bijection between the set of all homogeneous ideals~$J$ of~$P$ such that
$P/J$ is an Artinian Gorenstein algebra with socle degree~$d$ and
the set of all non-zero cyclic $P$-submodules $\langle \phi\rangle$ of~$\DD$
generated by a homogeneous element $\phi\in \DD_{-d}$.
\end{theorem}

Here the \textit{socle degree} $\sigma_S$ of a graded Artinian $\FF_q$-algebra $S=P/J$
is defined by $\sigma_S = \max \{ i\ge 0 \mid S_i \ne \{0\}\}$. In order to apply this theorem
to the Point Set Equivalence Problem, we still need to understand how the bijection
behaves with respect to homogeneous linear coordinate changes.

The general linear group $G=\GL_k(K)$ operates on $P = K[x_1,\dots,x_k]$ 
via homogeneous linear changes of coordinates,
i.e., the map $\ast:\; G \times P \longrightarrow P$ is given by $A \ast f(x_1,\dots,x_n) =
f(a_{11} x_1 + \cdots + a_{1k} x_k, \dots, a_{k1} x_1 + \cdots + a_{kk} x_k)$ for 
$A=(a_{ij}) \in G$ and $f\in P$. This group action defines a dual group action on~$\DD$ in the
natural way, namely as follows. Given a matrix $A\in G$, we let 
$\tau_A:\; P \longrightarrow P$ be given by $\tau_A(f)= A\ast f$. Then
the operation of~$A$ on~$\DD$ is defined by
$$
(A \ast \phi)(f) \;=\; (\phi\circ \tau_A)(f) \;=\; \phi( A\ast f) \hbox{\; for \;}
f\in P,\; \phi\in\DD.
$$
At this point we get the following result about the operation of~$G$ on the Macaulay
inverse system of an Artinian Gorenstein ring.

\begin{proposition}\label{prop:EquivarMacInv}
Let $J$ be a homogeneous ideal in $P = K[x_1,\dots,x_k]$ such that
$S=P/J$ is an Artinian Gorenstein ring. Let $\phi \in \DD_{-\sigma_S}$
be such that $J^\perp = \langle \phi\rangle$.
\begin{enumerate}
\item[(a)] Given a matrix $A\in \GL_k(K)$, let $J'=A\ast J
= \{A\ast f \mid f\in J\}$. Then we have $(J')^\perp = \langle A\ast\phi \rangle$.

\item[(b)] Given $j\ge 1$, a matrix $A\in \GL_k(K)$, and an element $\phi
\in \DD_{-j} \setminus \{0\}$, we have $\Ann_P(A\ast \phi) = A \ast \Ann_P(\phi)$.

\end{enumerate}
\end{proposition}

\begin{proof}
First we show~(a). Using the fact that the operation $\circ:\; P \times \DD
\longrightarrow \DD$ is $\GL_k(K)$-equivariant (see~\cite[Prop.~A.3]{IK}),
we calculate
\begin{align*}
  (J')^\perp &\;=\; \{\psi \in \DD \mid (A\ast f) \circ (A\, A^{-1}\ast \psi) = 0 
    \hbox{\; for all \;}f\in J\}\\
  &\;=\; \{ \psi\in \DD \mid A \ast (f \circ (A^{-1}\ast\psi)) = 0 
    \hbox{\; for all \;}f\in J\}\\
  &\;=\; \{ \psi \in \DD \mid f \circ (A^{-1}\ast\psi) = 0
    \hbox{\; for all \;}f\in J\}\\
  &\;=\; A \ast \{ A^{-1}\ast\psi \mid \psi\in\DD,\; f \circ (A^{-1}\ast\psi) = 0
    \hbox{\; for all \;}f\in J\}\\
  &\;=\; A \ast \langle \phi\rangle  \;=\; \langle A\ast\phi \rangle
\end{align*}

Claim~(b) follows similarly by calculating
\begin{align*}
\Ann_P(A\ast \phi) &\;=\; \{ f\in P \mid f \circ (A\ast\phi) = 0\}\\
  &\;=\; \{ f\in P \mid A\ast ((A^{-1}\ast f) \circ \phi = 0 \}\\
  &\;=\; A \ast \{ A^{-1}\ast f \mid f\in P,\; (A^{-1}\ast f) \circ \phi = 0 \}\\
  &\;=\; A \ast \Ann_P(\phi) \qedhere
\end{align*}
\end{proof}

Using this proposition, we can now return to our goal of reformulating the PSE problem.
Recall that, for a finite set of points $\X$ in $\PP^{k-1}$, the
doubling $D_\X = R / \JJ_{R/K[\ell]}$ of~$\X$ is an Artinian Gorenstein $K$-algebra
with socle degree $2r_\X-1$. We denote its Macaulay inverse system by $\II_\X$ and note that
it is of the form $\II_\X = \langle \Phi_\X\rangle$ with a homogeneous element 
$\Phi_\X \in \DD_{-2r_\X+1} = K[\pi_1,\dots,\pi_k]_{-2r_\X+1}$. If we consider 
$\pi_1,\dots,\pi_k$ as (standard graded) indeterminates, we can view~$\Phi_\X$ as a homogeneous polynomial
of degree $2r_\X-1$ in $K[\pi_1,\dots,\pi_k]$. It is called the \textit{Macaulay inverse polynomial}
of the set of points~$\X$. The polynomial $\Phi_\X$ can be found in the following way.

\begin{remark}{(Description of the Macaulay Inverse Polynomial)}\\
For a finite set of points $\X = \{p_1,\dots,p_n\}$ in $\PP^{k-1}$, the Buchberger-M\"oller Algorithm
allows us to compute the regularity index~$r_\X$ and representatives of the separators $f_1,\dots,f_n \in R_{r_\X}$
in~$P_{r_\X}$ in polynomial time. To compute $\Phi_\X$, we proceed as follows.
\begin{enumerate}
\item[(a)] By~\cite[Cor.~1.11]{Kre2}, the degree $2r_\X-1$ homogeneous
component of the canonical ideal $\JJ_{R/K[\ell]}$ consists of all elements
$\ell^{r_\X-1}\, (c_1f_1 + \cdots + c_n f_n)$ such that $c_1,\dots,c_n\in K$
and $c_1 + \cdots + c_n =0$. In other words, we have
$$
(\JJ_{R/K[\ell]})_{2r_\X-1} \;=\; K \cdot \ell^{r_\X-1}\, (f_2-f_1) \oplus \cdots \oplus
K \cdot \ell^{r_\X-1}\, (f_n-f_{n-1}) 
$$
Hence the residue class of $\ell^{r_\X-1}\, f_i$ is a $K$-basis of $(D_\X)_{2r_\X-1}$
for every $i\in\{ 1,\dots,n\}$.

\item[(b)] Next we let $\widehat{\JJ}_\X$ be a preimage of $\JJ_{R/K[\ell]}$ in~$P$,
and let $\hat{f_i} \in P_{r_\X}$ be a preimage of~$f_i$ for $i=1,\dots,n$.
Thus the above implies
$$
(\widehat{\JJ}_\X)_{2r_\X-1} \;=\; (I_\X)_{2r_\X-1} \;\oplus\;
K \cdot L^{r_\X-1}\, (\hat{f}_2 - \hat{f}_1) \oplus \cdots \oplus
K \cdot L^{r_\X-1}\, (\hat{f}_n - \hat{f}_{n-1}) 
$$
Therefore the projection $\Phi_\X: P_{2r_\X-1} \longrightarrow K$
to $L^{r_\X-1}\,\hat{f}_1$ along $(\widehat{\JJ}_\X)_{2r_\X-1}$ is a $K$-basis of $(\II_\X)_{-2r_\X+1}$,
i.e., it is the Macaulay inverse polynomial of~$\X$.
Notice that~$\Phi_\X$ is unique up to a scalar from~$K$ which depends on the choice
of a representative in~$P$ of a $K$-basis of $(D_\X)_{2r_\X-1}$.

\item[(c)] To get $\Phi_\X$ as a homogeneous polynomial of degree $2r_\X-1$
in $K[\pi_1,\dots,\pi_k]$, we may proceed as follows.
By multiplying the elements of a homogeneous reduced Gr\"obner basis of~$I_\X$
with terms of the complementary degree, we get a system of generators of the $K$-vector space
$(I_\X)_{2r_\X-1}$. Then we append the elements $\{ L^{r_\X-1}\, (\hat{f}_i - \hat{f}_i) \mid
i=2,\dots,n\}$ and interreduce the resulting list of polynomials $K$-linearly.
The result is a $K$-basis of the 1-codimensional vector subspace $(\widehat{\JJ}_\X)_{2r_\X-1}$
of~$P_{2r_\X-1}$.
Lastly, we use the bilinear pairing $\circ_{2r_\X-1}:\; P_{2r_\X-1} \times \DD_{-2r_\X+1}
\longrightarrow K$ defined before Proposition~\ref{prop:dividedpower}
to compute a basis $\{ \Phi_\X\}$ of the dual vector space.

\end{enumerate}
\end{remark}

Notice that, unfortunately, the method given in part~(c) of this remark has exponential
complexity, since $\dim_K(\widehat{\JJ}_\X)_{2r_\X-1} = \binom{2r_\X +k-2}{k-1} -1$.
The following result improves this method substantially.

\begin{theorem}[Computing the Macaulay Inverse Polynomial]\label{thm:CompMacInv}$\mathstrut$\\
Let $\X = \{ p_1,\dots, p_n\}$ be a finite set of points in~$\PP^{k-1}$,
let $L\in P_1$ be a linear form such that $\X \cap \mathcal{Z}(L)= \emptyset$, and let 
$\widehat{\JJ}_\X$ be the preimage of the canonical ideal $\JJ_R$ in~$P$.
\begin{enumerate}
\item[(a)] For a homogeneous polynomial $g\in P$ and $i=1,\dots,n$, define $g(p_i)$ by taking 
the value in $D_+(L)$. Then we have
$$
(\widehat{\JJ}_\X)_{2r_\X-1} \;=\; \{ g \in P_{2r_\X-1} \mid g(p_1) + \cdots + g(p_n) = 0 \}
$$

\item[(b)] For $i=1,\dots,n$, let $p_i = (1:p_{i1} : \dots : p_{ik})$ be the coordinate tuple
with respect to $(L,y_1,\dots,y_k)$ such that $L(p_i)=1$. Then we have
$$
\Phi_\X \;=\; \sum_{|\alpha| = 2r_\X-1} \left( \sum_{i=1}^n p_{i1}^{\alpha_1} \cdots p_{ik}^{\alpha_k} \right)
\pi_1^{\alpha_1} \cdots \pi_k^{\alpha_k}
$$
where the first sum extends over all $\alpha=(\alpha_1,\dots,\alpha_k) \in \NN^k$ such that
$|\alpha| = \alpha_1 + \cdots + \alpha_k = 2r_\X -1$.

\end{enumerate}
\end{theorem}

\begin{proof}
First we show~(a). Let $\ell$ be the residue class of~$L$ in~$R$. The formula for
$(\JJ_{R/K[\ell]})_{2r_\X-1}$ given in part~(a) of the preceding remark shows that
all polynomials $g\in (\widehat{\JJ}_\X)_{2r_\X-1}$ satisfy $g(p_1) + \cdots + g(p_n)=0$.
On the other hand, starting from $\HF_{D_\X}(2r_\X-1) = 1$, we see that $(\widehat{\JJ}_\X)_{2r_\X-1}$
has codimension~1 in $P_{2r_\X-1}$. This implies the claim.

Now we prove~(b). Let $g = \sum_\alpha c_\alpha \,x_1^{\alpha_1} \cdots x_k^{\alpha_k}$
be an element of $(\widehat{\JJ}_\X)_{2r_\X-1}$, where $c_\alpha\in K$ and the sum extends
over all $\alpha = (\alpha_1,\dots,\alpha_k) \in \NN^k$ such that $|\alpha| = 2r_\X-1$.
We write $\Phi_\X = \sum_\alpha b_\alpha\, \pi_1^{\alpha_1} \cdots \pi_k^{\alpha_k}$
with $b_\alpha \in K$. By construction of the Macaulay inverse polynomial, we know that
$g \circ \Phi_\X = \sum_\alpha c_\alpha b_\alpha = 0$ for all $g\in (\widehat{\JJ}_\X)_{2r_\X-1}$.
The formula in part~(a) yields
$$
g(p_1) + \cdots + g(p_n) \;=\; \sum_{|\alpha|=2r_\X-1} c_\alpha \; \left(
\sum_{i=1}^n p_{i1}^{\alpha_1} \cdots p_{ik}^{\alpha_k} \right) \;=\; 0
$$
Therefore, if we set $b_\alpha = \sum_{i=1}^n p_{i1}^{\alpha_1} \cdots p_{ik}^{\alpha_k}$
for all~$\alpha$ such that $|\alpha| = 2r_\X-1$, we obtain an element $\Phi_\X$
with the property $g\circ \Phi_\X =0$ for all $g \in (\widehat{\JJ}_\X)_{2r_\X-1}$.

As the Macaulay inverse polynomial is unique up to scalar multiples, it remains to show that
this element $\Phi_\X$ is non-zero. Since $\widehat{\JJ}_\X$ is properly contained in~$P_{2r_\X-1}$,
not all terms $t_\alpha = x_1^{\alpha_1} \cdots x_k^{\alpha_k}$ of degree $2r_\X-1$
satisfy $t_\alpha(p_1) + \cdots + t_\alpha(p_n)=0$. Hence at least one of the coefficients $b_\alpha$
is different from zero, and the proof is complete.
\end{proof}

Computing the formula in part~(b) of this theorem is again exponential in~$r_\X$ and~$k$.
However, this may not be as bad for point sets associated to ``good'' linear codes, as the following
remark shows.

\begin{remark}{(Regularity Bounds)}\label{rem:RegBound}\\
Let $K=\FF_q$, and let~$\X$ be a point set associated to a projective linear $[n,k,d]_q$-code~$C$
with a data rate $k/n$ which is at least $1/2$ and with a reasonably large minimal distance.
The minimal distance corresponds to the \textit{uniformity} of the point set~$\X$
(see~\cite{Kre2}, Prop.~5.11). For sets of points satisfying certain uniformity
conditions, bounds for the regularity index $r_\X$ are available (see~\cite{LP}
and~\cite[Sect.~2]{Kwa}).

For instance, the condition that any~$k$ points of~$\X$
span $\PP^{k-1}$ is commonly called \textit{linearly general position}. 
For points in linearly general position, the corollary
after~\cite[Thm.~5]{Kre1} yields $\Delta \HF_\X(i) \ge k-1$ for $1\le i\le r_\X-1$.
Combining this with our assumption on the data rate, 
we get $2k \ge n = \sum_{i=0}^{r_\X} \Delta \HF_\X(i) \ge (r_\X-1)(k-1) +2$, and therefore
$r_\X \le 3$. This means that we are applying our formula to polynomials of degree $2r_\X-1 \le 5$,
and since $\dim_{\FF_q}(P_5) = \binom{k+4}{5}$, the entire calculation is carried out in polynomial time.
In other words, if~$\X$ is the associated point set of a ``good'' linear code
with a reasonable data rate and minimal distance, the computation of the Macaulay inverse polynomial
can be effected in polynomial time.
\end{remark}

Finally, we are ready for the last contribution of this section.
A similar result in the local setting was shown in~\cite[Prop.~2.14]{Jel}.
Here we let $K=\FF_q$ again.

\begin{theorem}[Point Set Equivalence and Macaulay Inverse Polynomials]\label{thm:MacInvPoly}$\mathstrut$\\
Let $\X,\X'$ be two sets of $n$ $\FF_q$-rational points in $\PP^{k-1}$ each.
Assume that there exist linear forms $L,L'$ such that $\mathcal{Z}(L) \cap \X 
= \mathcal{Z}(L') \cap \X' = \emptyset$.
Let $R,R'$ be the homogeneous coordinate rings of~$\X$ and~$\X'$, respectively,
and let $\ell=L+I_\X \in R_1$ as well as $\ell' = L'+I_{\X'} \in R'_1$
be the residue classes of~$L$ and~$L'$.
Let $D_\X = R / \JJ_{R/\FF_q[\ell]}$ and $D_{\X'} = R' / \JJ_{R'/\FF_q[\ell']}$
be the doublings of~$\X$ and~$\X'$, respectively, and let
$\II_\X = \langle \Phi_\X\rangle$ and $\II_{\X'} = \langle \Phi_{\X'}\rangle$
be their Macaulay inverse systems.

Moreover, let $\Lambda:\; \PP^{k-1} \longrightarrow \PP^{k-1}$ be a linear change of
coordinates such that the induced map $\lambda: \FF_q[x_1,\dots,x_k] \longrightarrow
\FF_q[x_1,\dots,x_k]$ satisfies $\lambda(L)=L'$. Suppose that~$\Lambda$ is given by
a matrix $A\in \GL_n(\FF_q)$. Then the following conditions are equivalent.
\begin{enumerate}
\item[(a)] $\Lambda(\X)=\X'$

\item[(b)] $\Phi_\X = A^{-1}\ast \Phi_{\X'}$  in $\DD_{-2r_\X+1}$.

\end{enumerate}
\end{theorem}

\begin{proof}
By Theorem~\ref{thm:EquivDoubling}, condition~(a) is equivalent to
the fact that the homogeneous linear coordinate change defined by $A^{-1}$
induces an $\FF_q$-algebra isomorphism $\lambda_D:\, D_\X \longrightarrow D_{\X'}$.
In particular, the Hilbert functions of~$\X$ and~$\X'$ are equal in both cases
and the socle degrees of~$D_\X$ and~$D_{\X'}$ equal $2r_\X-1$.
Moreover, since~$A$ is unique up to a non-zero scalar multiple, we may assume
that $\lambda_D(f_1 + \dots + f_n) = f'_1 + \cdots + f'_n$. 
Thus Proposition~\ref{prop:EquivarMacInv} implies that condition~(a) is equivalent to~(b).
\end{proof}

In other words, this theorem says that the PSE problem
can be rephrased as the question whether the two homogeneous polynomials~$\Phi_\X$ 
and~$\Phi_{\X'}$ of degree $2r_\X-1$ differ only by a homogeneous linear
change of coordinates. This question is addressed in the next section.

\bigskip\bigbreak
%
%

\section{Point Set Equivalence and Polynomial Isomorphism}
\label{sec-8}

Once again, we let $P=\FF_q[x_1,\dots,x_k]$
be a polynomial ring over a finite field~$\FF_q$.
As in the preceding section,  we let the general linear group $G=\GL_k(\FF_q)$ 
operate on $P=\FF_q[x_1,\dots,x_k]$ via homogeneous linear changes of coordinates,
i.e., we let the map $\ast:\; G \times P \longrightarrow P$ be given by $A \ast f(x_1,\dots,x_n) =
f(a_{11} x_1 + \cdots + a_{1k} x_k, \dots, a_{k1} x_1 + \cdots + a_{kk} x_k)$ for 
$A=(a_{ij}) \in G$ and $f\in P$.

In this setting, we consider the third equivalence problem mentioned in the title of this paper.

\medskip
\noindent{\bf Polynomial Isomorphism Decision Problem.} Let $d\ge 1$, and let $f,g\in P_d$ be two
homogeneous polynomials of degree~$d$. Decide whether there exists a matrix $A\in \GL_k(\FF_q)$
such that $A \ast f = g$. In this case we say that~$f$ is \textit{equivalent} to~$g$ and write $f\sim g$.
\smallskip

\noindent{\bf Polynomial Isomorphism Search Problem.} Let $d\ge 1$, and let $f,g\in P_d$ be two 
equivalent homogeneous polynomials. Find a matrix $A\in \GL_k(\FF_q)$ such that $A \ast f = g$.
\smallskip

Sometimes the Polynomial Isomorphism (PI) problem is also called the \textit{polynomial equivalence problem}.
\medskip

The following main result of this section says that the PSE problem,
and hence the LCE problem, reduces in polynomial time to the 
computation of Macaulay inverse polynomials and the PI problem.

\begin{theorem}[PSE Problem Reductions]\label{thm:EquivReduce}$\mathstrut$\\
Let $C,C'$ be two projective $[n,k]_q$-codes in $\FF_q^n$ with generator matrices $G,G' \in \Mat_{k,n}(\FF_q)$,
respectively, let $\X,\X'$ be the associated point sets in $\PP^{k-1}$, and let $\Phi_{\X}, \Phi_{\X'}
\in \FF_q[\pi_1,\dots,\pi_k]$ be their Macaulay inverse polynomials.
\begin{enumerate}
\item[(a)] The PSE decision problem $\X \sim \X'$
reduces in polynomial time to the computation of~$\Phi_{\X},\, \Phi_{\X'}$ and the 
PI decision problem $\Phi_{\X} \sim \Phi_{\X'}$.

\item[(b)] The PSE search problem for $\X\sim \X'$
reduces in polynomial time to the computation of~$\Phi_{\X},\, \Phi_{\X'}$ and the
PI search problem for $\Phi_{\X} \sim \Phi_{\X'}$.

\end{enumerate}
\end{theorem}

\begin{proof}
Let us prove~(a) and~(b) simultaneously. Suppose that there exists a linear change of coordinates $\Lambda:\; \PP^{k-1}
\longrightarrow \PP^{k-1}$ such that $\Lambda(\X)=\X'$. Let $A\in \GL_k(\FF_q)$ be the matrix defining~$\Lambda$,
and let $\lambda:\; \FF_q[x_1,\dots,x_k] \longrightarrow \FF_q[x_1,\dots,x_k]$ be the linear change of coordinates
defined by~$A^{-1}$. By Theorem~\ref{thm:EquivDoubling}, the map~$\lambda$ induces an $\FF_q$-algebra
isomorphism $\lambda_D:\; D_\X \longrightarrow D_{\X'}$, and by Theorem~\ref{thm:MacInvPoly}, the
Macaulay inverse polynomials satisfy $\Phi_\X = A^{-1}\ast \Phi_{\X'}$  in $\DD_{-2r_\X+1}
=\FF_q[\pi_1,\dots,\pi_k]_{-2r_\X+1}$. 

Using Remark~\ref{rem:BuMo} and Proposition~\ref{prop:CompCanId}, we can calculate 
the vanishing ideals $I_\X,\, I_{\X'}$ 
and systems of generators of $\JJ_{R/\FF_q[\ell]}$ and $\JJ_{R'/\FF_q[\ell']}$ in polynomial time.
Thus we get systems of generators of the ideals $\widehat{J}_\X$ and $\widehat{J}_{\X'}$ in~$P$ 
which define~$D_\X$ and $D_{\X'}$ in polynomial time.

Now we compute the Macaulay inverse polynomials $\Phi_\X$ and $\Phi_{\X'}$ of these ideals.
They are homogeneous polynomials of degree $2r_\X-1$ in $\FF_q[\pi_1,\dots,\pi_k]$
where we consider $\pi_1,\dots,\pi_k$ as indeterminates and use the standard grading.
At this point we apply the hypothesis of~(a) and solve the PI decision problem 
for~$\Phi_{\X}$ and~$\Phi_{\X'}$. 

If the answer is affirmative, we use the hypothesis of~(b)
to find a matrix $A\in \GL_k(\FF_q)$ such that $\Phi_\X = A^{-1}\ast \Phi_{\X'}$.
Finally, we apply Theorem~\ref{thm:MacInvPoly} to conclude that the linear change of coordinates
$\Lambda:\; \PP^{k-1} \longrightarrow \PP^{k-1}$ given by~$A$ satisfies $\Lambda(\X)=\X'$.

If the answer is negative, we can use Theorem~\ref{thm:MacInvPoly} to conclude that the point
sets~$\X$ and~$\X'$ are not equivalent.
\end{proof}

In view of Proposition~\ref{prop:CEPandPSEP} and Remark~\ref{rem:CEPvsPSEP},
we have the analogous reductions of the LCE decision problem
and the LCE search problem. 
Notice that the PI search and decisions problems to which we reduce the PSE problems
involve homogeneous polynomials $\Phi_\X,\, \Phi_{\X'}$ of degree $2r_\X-1$ in~$k$
indeterminates. Known reductions of the LCE problem to the PI problem blow up the number
of indeterminates by a factor of $O(n^2)$ in the case of the method described in~\cite{AS1}, 
\cite{AS2} and~\cite{GQ1}, and by a linear factor in the case of the method in~\cite{GQ3}.

In the case of iso-dual codes,
we can improve these reductions substantially and get rid of the (possibly expensive)
computation of Macaulay inverse polynomials. This is our next, and final, topic.

\bigskip\bigbreak
%
%

\section{The Linear Code Equivalence Problem for Iso-Dual Codes}
\label{sec-9}

In this section we first consider a polynomial time reduction of the PSE search problem
to the PI search problem in the case of arithmetically Gorenstein point sets.
Then we deduce the analogous reduction for the LCE search problem for iso-dual codes.

Thus we let $P=\FF_q[x_1,\dots,x_k]$, we let $C\subseteq \FF_q^n$ be an indecomposable iso-dual
projective $[2k,k]_q$-code, we let $G\in \Mat_{k,2k}(\FF_q)$ be a generator matrix for~$C$,
and we let $\X = \{p_1,\dots,p_{2k}\}$ be the associated point set in $\PP^{k-1}$.
Recall that~$\X$ is self-associated, arithmetically Gorenstein, and the difference function of its Hilbert function
satisfies $\Delta \HF_\X:\; 1\;\; k{-}1 \;\; k{-}1 \;\; 1$.

Moreover, we assume that there is a linear form $L\in P_1$ such that $\X \cap \mathcal{Z}(L) = \emptyset$.
Its image $\ell \in R_1$ in the homogeneous coordinate ring $R=P/I_\X$ of~$\X$ is a non-zerodivisor
and the Artinian reduction $\overline{R} = R / \langle \ell \rangle$ is a 0-dimensional
Gorenstein ring with Hilbert function $\Delta \HF_\X$. In particular, the socle degree of~$\overline{R}$
is $\sigma_{\bar{R}} = 3$.
By Theorem~\ref{thm:MacInverse}, the Macaulay inverse system of~$\overline{R}$ is generated by a homogeneous element
$\phi_\X\in \DD_{-3}$. In analogy to Theorem~\ref{thm:MacInvPoly}, we have the following result.

\begin{proposition}\label{prop:ArithGorMacInverse}
Let $\X,\X'$ be two arithmetically Gorenstein sets of $\FF_q$-rational points in $\PP^{k-1}$
such that $\Delta\HF_\X = \Delta \HF_{\X'}:\; 1\; k{-}1 \; k{-}1 \; 1$. Let $\phi_\X,\phi_{\X'}
\in \DD_{-3}$ be the homogeneous elements generating the Macaulay inverse systems of~$\overline{R}$
and $\overline{R}'$, respectively.

Moreover, let $\Lambda:\; \PP^{k-1} \longrightarrow \PP^{k-1}$ be a linear change of
coordinates such that the induced map $\lambda: \FF_q[x_1,\dots,x_k] \longrightarrow
\FF_q[x_1,\dots,x_k]$ satisfies $\lambda(L)=L'$. Suppose that~$\Lambda$ is given by
a matrix $A\in \GL_n(\FF_q)$.

If $\Lambda(\X)=\X'$, then $\phi_\X = A^{-1}\ast \phi_{\X'}$  in $\DD_{-3}$.
\end{proposition}

\begin{proof}
Clearly, if $\Lambda(\X)=\Lambda(\X')$ then the homogeneous linear coordinate change defined by $A^{-1}$
induces an $\FF_q$-algebra isomorphism $\lambda_R:\, R \longrightarrow R'$.
Let $\ell,\ell'$ be the residue classes of~$L$ and~$L'$ in~$R$ and~$R'$, respectively, and let
$\overline{R} = R/\langle \ell\rangle$ as well as $\overline{R}' = R'/\langle \ell'\rangle$.
Since the map~$\lambda_R$ satisfies $\lambda_R(\ell)=\ell'$, it induces an $\FF_q$-algebra
isomorphism $\overline{\lambda}:\; \overline{R} \longrightarrow \overline{R}'$.

As~$\X$ and~$\X'$ are assumed to be arithmetically Gorenstein, the rings~$\overline{R}$
and~$\overline{R}'$ are Artinian Gorenstein rings. They have both socle degree~3.
By Theorem~\ref{thm:MacInverse}, their Macaulay inverse systems~$\DD$ and~$\DD'$
are generated by homogeneous elements $\phi_\X\in \DD_{-3}$ and $\phi_{\X'}\in \DD'_{-3}$.
Thus Proposition~\ref{prop:EquivarMacInv} implies the claim.
\end{proof}

Together with the following result, this proposition yields a polynomial time reduction
from the PSE seach problem to the 3-PI search problem, as explained below.
The explicit computation of the Macaulay inverse polynomial of the Artinian reduction
of an arithmetically Gorenstein set of points has been described before in~\cite{EI},
and when the base field has characteristic zero in more detail in~\cite{Toh} and~\cite{ER}. 
The next proposition may be seen as a slightly improved 
characteristic free version of these results.

\begin{proposition}[Computing the Macaulay Inverse Polynomial of an Artinian 
Reduction of an Arithmetically Gorenstein Set of Points]\label{prop:CompMIPofAR}$\mathstrut$\\
Let $K$ be a field, let $\X=\{p_1,\dots,p_n\}$ be a set of $K$-rational points in $\PP^{k-1}$,
and assume that $x_1\in R_1$ is a non-zerodivisor. Let $\overline{R} = R/ \langle x_1\rangle$ 
and write $p_i = (1: p_{i2} : \cdots : p_{ik})$ with $p_{ij}\in K$ for $i=1,\dots,n$.
For $i=1,\dots,n$, let $L_i = \pi_1 + p_{i2} \pi_2 + \cdots + p_{ik} \pi_k \in \DD_{-1}$.

Assume that~$\X$ is arithmetically Gorenstein.
Let $f_1,\dots,f_n\in R_{r_\X}$ be the separators of~$\X$ and write $\bar{f}_j = \beta_j \, \bar{f}_1$
with $\beta_j\in K\setminus \{0\}$ for $j=2,\dots,k$, where $\bar{f}_i$ is the residue class of~$f_i$
in $\overline{R}_{r_\X}$.

In this setting, a Macaulay inverse polynomial $\phi_\X$ of~$\X$ is given by
$$
\phi_\X \;=\;  L_1^{[r_\X]} + \beta_2 L_2^{[r_\X]} + \cdots + 
\beta_n L_n^{[r_\X]} \in \DD_{-r_\X} 
$$
Here we let $L_i^{[j]} = \sum_{|\alpha|=j} p_i^\alpha \pi^\alpha$ for all $i=1,\dots,n$ and $j\ge 0$.
\end{proposition}

\begin{proof} 
By~\cite[Thm.~I]{EI}, there exists a tuple $(\alpha_1,\dots,\alpha_n) \in K^n$ such that
$\phi_\X = \alpha_1 L_1^{[r_\X]} + \cdots + \alpha_n L_n^{[r_\X]}$. Now we use the formula
$h\circ L_i^{[j]} = L_i^{[j-r]}\cdot h(p_i)$ for $h\in R_r$ given in~\cite[Sec.~2, Lemma]{EI}.
It shows that $f_i\circ L_j^{[r_\X]} = \delta_{ij}$ for all $i,j$, and therefore 
$f_i\circ \phi_\X = \alpha_i$ for $i=1,\dots,n$.

As shown in~\cite[Prop.~2.2]{Toh}, we have $x_1\circ \phi_\X = \alpha_1 L_1^{[r_\X-1]} + \cdots +
\alpha_n L_n^{[r_\X-1]} = 0$. (The assumption ${\rm char}(K)=0$ made there plays no role in the proof.)
Consequently, if we write $f_j - \beta_j f_1 = x_1 g_j$ with $g_j \in R_{r_\X-1}$ for $j=2,\dots,n$,
then 
$$
0 \;=\; g_j \circ (x_1 \circ \phi_\X) \;=\; (x_1 g_j) \circ\phi_\X = 
(f_j - \beta_j f_1) \circ \phi_\X \;=\; (\alpha_j -\beta_j \alpha_1) \circ \phi_\X
$$
shows $\alpha_j=\beta_j \alpha_1$ for $j=2,\dots,n$. As the tuple $(\alpha_1,\dots,\alpha_n)$
is determined only up to a non-zero scalar factor, this proves the claim.
\end{proof}

Given a bound for~$r_\X$, the polynomial $\phi_\X$ can be calculated efficiently. This results in the
following polynomial time reduction.

\begin{corollary}\label{cor:ArithGorPEQ}
The PSE search problem for arithmetically Gorenstein point sets
reduces in polynomial time to the PI search problem in degree~3.
\end{corollary}

\begin{proof}
In the proof of Prop.~\ref{prop:ArithGorMacInverse}, the ideals~$I_\X + \langle L\rangle$ and $I_{\X'} +
\langle L'\rangle$ defining~$R$ and~$R'$, respectively, can be found in polynomial
time by Remark~\ref{rem:BuMo}. Then vector space bases of their homogeneous components of degree~3
can be found in polynomial time, as $\dim_{\FF_q}(P_3) = \binom{k+2}{3}$.
Finally, the calculation of the Macaulay dual polynomials $\phi_\X, \phi_{\X'} \in \FF_q[\pi_1,\dots,\pi_k]_3$
is achieved in polynomial time using the preceding proposition.
\end{proof}

The reduction described in this corollary has several advantages over the one in Theorem~\ref{thm:EquivReduce}:
we do not need to compute Macaulay inverse polynomials, and the resulting PI search problem is for
homogeneous polynomials of degree~3 (rather than degree $2r_\X-1=5$).

Notice that it is not clear how to reduce the PSE decision problem
in this manner, because if we have a matrix $A^{-1} \in \GL_k(\FF_q]$ which defines an
$\FF_q$-algebra isomorphism $\overline{\lambda}:\; \overline{R}\longrightarrow \overline{R}'$,
we do not know whether the corresponding linear change of coordinates lifts to a map 
$\lambda_R:\; R \longrightarrow R'$, i.e., whether $\lambda(I_\X) = I_{\X'}$.

As a consequence of the corollary, the LCE search problem for 
indecomposable iso-dual codes reduces in polynomial time as follows.

\begin{proposition}[Reduction of the LCE Problem for Iso-Dual Codes]\label{thm:IsoDualReduce}$\mathstrut$\\
Let $C,C'$ be two indecomposable iso-dual projective $[2k,k]_q$-codes in $\FF_q^n$ with generator matrices 
$G,G' \in \Mat_{k,n}(\FF_q)$, respectively. Let $\X,\X'$ be the associated point sets in $\PP^{k-1}$, 
and let $\phi_{\X}, \phi_{\X'} \in \FF_q[\pi_1,\dots,\pi_k]_3$ be the polynomials generating
their Macaulay inverse systems.

Then the LCE search problem for $C\sim C'$ reduces in polynomial time
to the PI search problem $\phi_\X \sim \phi_{\X'}$ in degree~3.
\end{proposition}

\begin{proof}
It suffices to combine the polynomial time reduction of Corollary~\ref{cor:ArithGorPEQ}
with Proposition~\ref{prop:CEPandPSEP} and Remark~\ref{rem:CEPvsPSEP}.
\end{proof}

\bigskip\bigbreak
%
%

\section{Conclusion and Outlook}\label{sec-10}

In this paper we have studied the difficulty of the Linear Code Equivalence (LCE) problem
using techniques from algebraic geometry. We showed that the LCE problem is polynomial time equivalent
to the Point Set Equivalence (PSE) problem for finite sets of $\FF_q$-rational points in~$\PP^{k-1}$.
Using the doubling of such a set~$\X$, we proved that the LCE problem is equivalent to an algebra isomorphism
problem for Artinian Gorenstein algebras. Then we applied the Macaulay inverse system
to reduce the task to a Polynomial Isomorphism (PI) problem. 
Although this reduction cannot be performed in polynomial time in general, it can if we assume that~$C$
is a reasonably ``good'' code. The PI problem has been studied
in Algebraic Geometry for more than a century, and in Cryptography for at least the last 30 years.
For iso-dual codes, we have reduced the LCE problem in polynomial time to a PI
problem in degree~3 which has been analyzed extensively in Cryptography under the name IP1S
and is believed to be efficiently solvable in the vast majority of cases.

The next step is obviously to implement the reductions presented here and to check whether the
purported time complexities hold up in practice.  
Furthermore, it is an interesting problem to describe which homogeneous polynomials of some odd degree
generate the Macaulay inverse systems of doublings of point sets, in order to understand whether the
PI problem is actually harder than the LCE problem in certain cases.

\section*{Declarations}
No funding was received to assist with the preparation of this manuscript.
The author has no relevant financial or non-financial interests to disclose.

\bigbreak
%
%

\end{document}